\documentclass[lettersize,onecolumn]{IEEEtran}
\usepackage{authblk}

\usepackage{color}
\usepackage{enumerate}
\usepackage{verbatim}
\usepackage{amsthm,amssymb}
\usepackage{amsmath}
\usepackage{thmtools,thm-restate}
\usepackage{wrapfig}
\usepackage{xspace}
\usepackage{amsfonts}
\usepackage{bbm}
\usepackage[noend]{algpseudocode}
\algnewcommand{\UnindentedComment}[1]{\Statex \hspace*{\algorithmicindent} \(\triangleright\) #1}
\usepackage{algorithm}

\usepackage{todonotes}

\definecolor{darkgreen}{rgb}{0,0.5,0}
\usepackage{hyperref}
\hypersetup{
    unicode=false,          % non-Latin characters in AcrobatÎ“Ã‡Ã–s bookmarks
    colorlinks=true,        % false: boxed links; true: colored links
    linkcolor=blue,          % color of internal links (change box color with linkbordercolor)
    citecolor=darkgreen,        % color of links to bibliography
    filecolor=magenta,      % color of file links
    urlcolor=cyan           % color of external links
}

\usepackage[disableredefinitions,roman]{complexity}
\usepackage[framemethod=tikz]{mdframed}
\global\mdfdefinestyle{myframe}{leftmargin=.75in,rightmargin=.75in,linecolor=black,linewidth=1.5pt,innertopmargin=10pt,innerbottommargin=10pt} 
\usepackage{mdframed}
\usepackage{framed}
\usepackage{nicefrac}

\date{}

\newtheorem{theorem}{Theorem}[section]
\newtheorem{lemma}[theorem]{Lemma}
\newtheorem{claim}[theorem]{Claim}

\newtheorem{definition}[theorem]{Definition}
\newtheorem{corollary}[theorem]{Corollary}

\usepackage[capitalize, nameinlink]{cleveref}
\crefname{theorem}{Theorem}{Theorems}
\Crefname{lemma}{Lemma}{Lemmas}
\Crefname{alg}{Algorithm}{Algorithms}
\Crefname{claim}{Claim}{Claims}
\Crefname{infclaim}{Claim}{Claims}
\Crefname{observation}{Observation}{Observations}
\Crefname{invariant}{Invariant}{Invariants}
\Crefname{algorithm}{Algorithm}{Algorithms}

\newcommand{\eps}{\varepsilon}

\renewcommand{\G}{\mathcal{G}_{\eta}}
\newcommand{\match}{\operatorname{Match}}
\newcommand{\sol}{\operatorname{SOL}}
\newcommand{\opt}{\operatorname{OPT}}
\newcommand{\ind}{\operatorname{Ind}}

\newcommand{\rtag}[1]{\text{\quad $\left(\text{#1}\right)$}}

\usepackage{amsmath,amsfonts}
\usepackage{algorithm}
\usepackage{array}
\usepackage[caption=false,font=normalsize,labelfont=sf,textfont=sf]{subfig}
\usepackage{textcomp}
\usepackage{stfloats}
\usepackage{url}
\usepackage{verbatim}
\usepackage{graphicx}
\usepackage{cite}
\begin{document}

\title{Efficient $\eps$-approximate minimum-entropy couplings}
\author[]{Spencer Compton}
\affil[]{Stanford University \authorcr
  comptons@stanford.edu}

\maketitle

\begin{abstract}
Given $m \ge 2$ discrete probability distributions over $n$ states each, the minimum-entropy coupling is the minimum-entropy joint distribution whose marginals are the same as the input distributions. Computing the minimum-entropy coupling is strongly NP-hard, but there has been significant progress in designing approximation algorithms; prior to this work, the best known polynomial-time algorithms attain guarantees of the form $H(\operatorname{ALG}) \le H(\operatorname{OPT}) + c$, where $c \approx 0.53$ for $m=2$, and $c \approx 1.22$ for general $m$ \cite{compton2023minimum}.

A main open question is whether this task is APX-hard, or whether there exists a polynomial-time approximation scheme (PTAS). In this work, we design an algorithm that produces a coupling with entropy $H(\operatorname{ALG}) \le H(\operatorname{OPT}) + \eps$ in running time $n^{\poly(1/\eps) \cdot \exp(O(m)) }$: showing a PTAS exists for constant $m$.
\end{abstract}


\section{Introduction}
The minimum-entropy coupling problem entails the following task: given $m \ge 2$ discrete probability distributions $p_1,\dots,p_m$ over $n \ge 2$ states each, what is the joint distribution of minimum Shannon entropy whose marginals are the same as the input distributions? For $m=2$, this can alternatively be viewed as a coupling that maximizes mutual information, since $I(X;Y) = H(X)+H(Y)-H(X,Y)$. 

While \textit{maximizing} the entropy of the coupling is simple (the joint distribution where variables are independent), \textit{minimizing} the entropy is a concave minimization problem that is typically more challenging (e.g. \cite{cardinal2008tight}). Indeed, the minimum-entropy coupling problem is strongly NP-hard \cite{kovavcevic2012hardness}, yet is a fundamental task in information theory with numerous applications that motivate the study of polynomial-time approximation algorithms. 

\subsection{Related applications}

In communications, suppose we want to send information about some random variable $X$, yet we must do so by communicating in a manner with marginal distribution $Y$. In this case, we may proceed by communicating according to some chosen joint distribution over $X,Y$ (i.e., given a realization of $X=x$, then communicate corresponding to the conditional distribution of $Y \, | \, X=x$); the minimum-entropy coupling of $X,Y$ is exactly the valid communication procedure that maximizes the mutual information $I(X;Y)$. This is realized by \cite{sokota2022communicating}, who leverage reinforcement learning and the coupling algorithm of \cite{cicalese2019minimum} to play \textit{Markov coding games}, where they hope to successfully play a game while also communicating through Markov decision process trajectories. For example, they design an agent that simultaneously plays the game Pong while also communicating images via its actions in the game Pong. The work of \cite{ebrahimi2024minimum} introduces a variant \textit{minimum-entropy coupling with bottleneck}, which they also employ for Markov coding games.

The work of \cite{de2022perfectly} similarly leverages this coupling communication perspective for \textit{steganography}, where one hopes to encode secret messages in innocuous-seeming text. They use the greedy coupling algorithm of \cite{kocaoglu2017entropic} to encode secret messages in text (GPT-2), audio (WaveRNN), and images (Image Transformer). Later work of \cite{sokota2024computing} uses heuristics to construct low-entropy couplings of autoregressive distributions, so they may perform steganography when message priors are large/autoregressive.

In causal inference, we often hope to determine the direction of causal relationships between random variables. The \textit{entropic causal inference} framework of \cite{kocaoglu2017entropic} aims to learn causal directions from only observational data, by fitting in the ``simpler'' direction requiring less entropy (motivated by Occam's razor). In the pairwise setting \cite{kocaoglu2017entropic,kocaoglu2017isit,compton2020entropic}, the cost of fitting $X \rightarrow Y$ is $H(X)+H(E)$ for the minimum entropy $E$ where $X \perp\!\!\!\!\perp E$, and there exists some function $f$ where when $Y=f(X,E)$, then $X,Y$ have the correct joint distribution; the minimum entropy $E$ is ultimately given by the minimum-entropy coupling of all conditional distributions $Y \, | \, X=x$. This line of work is also extended to learning entire causal graphs \cite{compton2022entropic}, and further related work \cite{javidian2021quantum,javidian2022quantum}.

We now more briefly discuss a wider collection of applications. \cite{liang2023multimodal} use minimum-entropy couplings to estimate interactions between different modalities in multimodal data where labeling is time-consuming. \cite{bounoua2025learning} leverage a continuous version of the minimum-entropy problem to match multimodal data. \cite{chowdhury2025fundamental} uses minimum-entropy couplings for \textit{concept erasure}, where we hope to erase information about a sensitive attribute while retaining maximum information. A line of work in \cite{zamani2024improving,zamani2025private,zamani2025variable} uses minimum-entropy couplings for private data compression. The works of \cite{vidyasagar2012metric,cicalese2016approximating} use minimum-entropy couplings for dimensionality reduction. Further applications are generally well-discussed in \cite{cicalese2019minimum,li2021efficient}, including functional representation in \cite{cicalese2019minimum}, and random number generation in \cite{li2021efficient}.

\subsection{Prior algorithmic guarantees}
Since computing the minimum-entropy coupling is strongly NP-hard, there has been much interest in designing approximation algorithms where the entropy of the algorithm's output coupling, $H(\operatorname{ALG})$, is not much larger than that of the optimal coupling, $H(\operatorname{OPT})$. Let us call an algorithm $c$-additive if it produces a coupling where $H(\operatorname{ALG}) \le H(\operatorname{OPT}) + c$. 

The first works on approximation guarantees studied a collection of algorithms with respect to a \textit{majorization} lower bound. Cicalese, Gargano, and Vaccaro \cite{cicalese2016approximating} introduced an algorithm that is $1$-additive when $m=2$, and $\lceil \log(m) \rceil$-additive for general $m$: yielding the first constant-additive guarantee for constant $m$. The work of Kocaoglu, Dimakis, Vishwanath, and Hassibi \cite{kocaoglu2017entropic} introduced the \textit{greedy coupling algorithm}, that was later shown to be a local optimum by the same authors \cite{kocaoglu2017isit}, and a $1$-additive algorithm when $m=2$ by Rossi \cite{rossi2019greedy}. Later, Li \cite{li2021efficient} introduced a new algorithm that is $(2-2^{2-m})$-additive: yielding the first additive constant for general $m$. Compton \cite{compton2022tighter} improved the greedy coupling guarantee to $\log(e) \approx 1.44$ for general $m$, while also showing a barrier in how the approximation analysis was tight with respect to the majorization lower bound.

The simultaneous works of Compton, Katz, Qi, Greenewald, and Kocaoglu \cite{compton2023minimum}, and Shkel and Yadav \cite{shkel2023information}, introduce a stronger lower bound called the \textit{profile} lower bound (or the \textit{information spectrum}, respectively). Moreover, in \cite{compton2023minimum}, they improve the guarantees of the greedy coupling algorithm: showing it is $\log(e)/e \approx 0.53$-additive for $m=2$, and $(1+\log(e))/2 \approx 1.22$-additive for general $m$ (see Appendix D of \cite{compton2023minimum} for guarantees in other small values of $m$). Prior to this work, these are the best-known approximation guarantees for polynomial-time algorithms.

More tangentially related, the work of \cite{compton2023minimum} also provides algorithms for exactly computing the minimum-entropy coupling in exponential time. Additionally, the works of \cite{yadav2025information,ma2025efficient} argue how the profile-based proof techniques can also give approximation guarantees for general R\'enyi entropy. 

Summarizing the current landscape, previously there was no known polynomial-time algorithm for coupling $m=2$ distributions where $H(\operatorname{ALG}) \le H(\operatorname{OPT}) + 0.52$, and it seems natural to wonder whether there exists some constant $c^*>0$ such that no polynomial-time algorithm can attain $H(\operatorname{ALG}) \le H(\operatorname{OPT}) + c^*$ (i.e., the problem is APX-hard). In this work, we will design a new algorithm that refutes this possibility: showing a PTAS when $m$ is a constant.
\subsection{Our result}

Let $\opt$ be a minimum-entropy coupling for $m$ discrete distributions with at most $n$ distribution states each. In our main result, we show the existence of an efficient $\eps$-additive algorithm:
\begin{restatable}{theorem}{maintheorem}\label{theorem:main-theorem}
     For any $0 < \eps < \nicefrac{1}{2}$, there exists an algorithm with running time $n^{O(m^7 \cdot 2^{6m} \cdot \log^4(1/\eps)/\eps^2 )}$ that outputs a coupling $\operatorname{ALG}$, where $H(\operatorname{ALG}) \le H(\opt) + \eps$.
\end{restatable}

As an immediate corollary, we conclude how for constant $m$ and constant $\eps>0$, the running time is polynomial in $n$. This significantly improves beyond prior works, where even for $m=2$ there was no known polynomial-time algorithm with additive approximation guarantee smaller than $\log(e)/e \approx 0.53$.
\begin{corollary}
    For any constant $m=O(1)$ and constant $\eps > 0$, there exists an algorithm with running time $n^{O(1)}$ that outputs a coupling $\operatorname{ALG}$, where $H(\operatorname{ALG}) \le H(\opt) + \eps$.
\end{corollary}

In \cref{sec:algo} we will introduce our new algorithm, in \cref{sec:analysis} we will analyze its approximation error, and in \cref{sec:discuss} we will discuss remaining open problems.

\subsection{Preliminaries}
Throughout this paper, $\log$ is base 2. The value of $H(x)$ refers to the Shannon entropy of $x$, and is applied to any non-negative vector or multiset (meaning, $H(x)=\sum_i x_i \log(1/x_i)$). For any non-negative scalar $x$, we similarly use $\phi(x) = x \log(1/x)$.

\section{$\eps$-approximate coupling algorithm}\label{sec:algo}

We begin by sharing some motivation for our approach. Suppose you are trying to solve an optimization problem where the input is a collection of numbers in a bounded range. A classical idea in approximation algorithms is to reduce the number of distinct values of input numbers (say, by rounding numbers to some power of $(1+\eps)$, and arguing how sufficiently small numbers can be handled separately), and then design an algorithm with running time that depends on the distinct number of values. For example, this technique was famously employed for the makespan minimization problem by Hochbaum and Shmoys \cite{hochbaum1987using} and later work (e.g. \cite{leung1989bin,hochba1997approximation,alon1997approximation,alon1998approximation,jansen2010eptas,chen2014optimality,jansen2020closing}). In the makespan minimization problem, $n$ jobs with different completion times must be assigned among $m$ identical machines, and the goal is to minimize the maximum load given to a machine. Without giving much detail, the makespan minimization literature is a natural place to look for ideas, as it is quite similar to the 3-Partition problem from which \cite{kovavcevic2012hardness} showed strong NP-hardness of minimum-entropy coupling. 

Unfortunately, there are some clear obstacles towards using this distinct-values approach for minimum-entropy coupling. Primarily, the range of important numbers is too large. After renormalizing, makespan minimization can focus on values in some range $[\eps,1]$, but for minimum-entropy coupling there seem to be important probability states with values in a much larger range of $[\poly(\eps/n),1]$. Upon first glance, this is what looks quite unsalvageable for using this style of approach. As far as we can tell, any naive application of the ideas in the makespan minimization literature cannot attain a running time with polynomial dependence in $n$ for our problem.  There are also less intimidating (but important) obstacles for constructing couplings with this approach; for example, if we couple $z$ mass from two distributions' states $p_i(a),p_j(b)$, then the remaining masses of $p_i(a)-z$ and $p_j(b)-z$ must remain in our collection of distinct values. 

Despite these obstacles, we will still be able to design a PTAS leveraging a strategy that bounds the number of distinct values. We will begin with a rough sketch of an approach that is too slow, then refine it to a sketch of an approach that will be sufficiently fast, and finally provide a rigorous algorithm. For simplicity, the first two sketches will focus on the case of coupling just two distributions $p_1,p_2$. 

\textbf{Sketch of slow approach. } Our first approach will be a dynamic program where we maintain the rounded values of the remaining uncoupled mass for each distribution, and create couplings one state at a time. We will make some strong assumptions in the presentation of this slow approach (some assumptions are not quite correct), but they will help us convey intuition: (i) suppose there is some $\tau$ where we may ignore any distribution state once it is smaller than $\tau$ (you should think of these states as later being handled by some postprocessing), (ii) suppose there is a set $\mathcal{G}$ of values in $[\tau,1]$, where it is always fine  to round down any distribution state value to $\mathcal{G}$ throughout the process of coupling, and (iii) when choosing the size of the next coupling state, it is fine to only consider values in $\mathcal{G}$.\footnote{In reality, this second assumption is particularly incorrect (it would result in uncoupled mass), but this improves presentation for the sketch.} Since we ignore states with value $< \tau$, then at every point in the coupling process there are at most $1/\tau$ relevant distribution states in each of $p_1$ and $p_2$. Since each distribution state takes one of $|\mathcal{G}|$ values, this means the relevant information about $p_1,p_2$ can be described as one of $(1 + 1/\tau)^{2|\mathcal{G}|}$ dynamic programming states (for the rest of the paper, we call these DP-states). When choosing the next state of the coupling, there are $|\mathcal{G}|$ options for what size to make the coupling state (via assumption (iii)), $|\mathcal{G}|$ options for what size of distribution state from $p_1$ will be coupled, and $|\mathcal{G}|$ options for what size of distribution state from $p_2$ will be coupled, for a total of $O(|\mathcal{G}|^3)$ possible transitions. Consider briefly the pseudocode in \cref{alg:slow}, which would not actually produce a valid coupling, but gives a sense of an initial approach one might hope to use.\footnote{Do not read much into the base case of this dynamic program, or generally the details in the pseudocode. The main purpose is to see a rough outline of how a relevant dynamic program might look, before we introduce more technical ideas.}

\begin{algorithm}[]
    \caption{Slow algorithm sketch (does not actually produce a valid coupling)}
   \label{alg:slow}
\begin{algorithmic}[1]
    \State {\bfseries Input:} Multisets $S_1,S_2$ of elements of $\mathcal{G}$ detailing the remaining distribution states of $p_1,p_2$.
    \Procedure{SlowCoupling}{$S_1,S_2$} 
    \If{DP[$S_1,S_2$] already computed}
        \Return DP[$S_1,S_2$]
    \EndIf
    \If{$S_1 = \emptyset$ or $S_2 = \emptyset$}
        \Return 0
    \EndIf
    \State DP[$S_1,S_2$] $\gets \infty$
    \For{$z \in \mathcal{G}$, $a \in S_1$, $b \in S_2$, $z \le a,b$} 
    \UnindentedComment{Consider coupling $z$ mass from a state of size $a$ from $p_1$ and size $b$ from $p_2$}
    \State DP[$S_1,S_2$] $\gets \min(\text{DP}[S_1,S_2], \phi(z) + \operatorname{SlowCoupling}((S_1 \backslash \{a\}) \cup \operatorname{Round}(a-z), (S_2 \backslash \{b\}) \cup \operatorname{Round}(b-z))$
    \EndFor
    \State \Return DP[$S_1,S_2$].
    \EndProcedure
\end{algorithmic}
\end{algorithm}

In addition to the above algorithm not actually producing a valid coupling, the main concern is how the running time is about $(1/\tau)^{O(|\mathcal{G}|)}$, where we roughly expect the parameters to take values like $\tau = \poly(\eps/n)$, and $|\mathcal{G}| =  O(\log(1/\tau)/\eps)$.

\textbf{Motivation for a faster approach. } We would like to improve the running time by reducing the size of the DP-state space. Ideally, we would like for the algorithm to somehow only need to maintain the number of distribution states of $p_1,p_2$ taking values in some small range $[l,r] \cap \mathcal{G}$, where $[l,r]$ is a sliding window.  More concretely,  our new dynamic program will try to maintain an invariant where for a range $[l,r]$ specified in its DP-state: (i) all remaining distribution states $\ge \tau$ take values in $\mathcal{G}$, (ii) there are no remaining distribution states $>r$, (iii) we track the counts of the number of distribution states with each value in $[l,r] \cap \mathcal{G}$ for $p_1,p_2$, and (iv) all distribution states $<l$ have not been modified.

For a fixed $[l,r]$, the DP-state is then entirely defined by the counts of the number of distribution states with each value in $[l,r] \cap \mathcal{G}$ for $p_1,p_2$. If it is feasible to use a sufficiently small sliding window, then this will let us get our desired running time. 

However, maintaining such an invariant (while still being approximately optimal) seems quite daunting. For example, consider the case where $p_1$ is uniform over $n/k$ states, and $p_2$ is uniform over $n$ states, for some large value of $k$. We will want to couple distribution states of $p_1$ with distribution states of $p_2$, but these distribution states are a large factor of $k$ apart, so they would not simultaneously be within a small sliding window. One of our key observations will be to remedy this by introducing the option of \textit{splitting the largest distribution state in our sliding window in half} whenever the dynamic program chooses not to couple it with some distribution state of similar mass. In our earlier example, our intuition is if we keep splitting distribution states of $p_1$ in half until they fit into the same sliding window as $p_2$ distribution states, then later coupling the distribution states of $p_1$ and $p_2$ will not introduce much approximation error. At this moment, it should be very non-obvious why the splitting approach is not too lossy in general, but we will defer analysis until later. With the main motivations behind our faster algorithm in hand, in addition to our vague fix of splitting to try keep everything within a small sliding window, we now introduce the full algorithm.

\textbf{Introducing the algorithm. } In the remainder of this section, we will describe our algorithm and why it finds a valid coupling, but will defer the proof of its approximation error. Recall the input is $m$ distributions $p_1,\dots, p_m$, each with at most $n$ distribution states. Our algorithm will use an internal parameter $\eta$, which plays a role similar to $\eps$. Later, we will bound the approximation error in terms of $\eta$, and choosing the value of $\eta$ in terms of $\eps$ will enable an $\eps$-additive guarantee.

\textit{\underline{Rounding points $\G$.}}  Throughout this work, a main goal is to have a bounded number of distinct values of distribution states. Originally, distribution states can take any value in $[0,1]$. Instead, we will restrict to a set of points $\G \subset (0,1]$ where, after rounding the initial distribution states to these points, we hope to be able to stay within this set. For example, we would like for it to be the case that $|\G|$ is ``small'', and for any $x,y \in \G$ where $x > y > 0$, that $x-y \in \G$. In reality, we will leverage more nuanced properties, like how if $x,y$ are close (but unequal), then $x-y$ can be written as the sum of two elements in $\G$ that are not too small. We define the points of $\G$:
\begin{definition}
    Suppose $\eta \le \frac{1}{2^{3m} \cdot 4m} \le \frac{1}{512}$ is a power of 2. Then, 
    \begin{equation*}
    \G \triangleq \{2^{-i}\eta j \, \, \, | \, \, \, i \in \mathbb{Z}, \, \, \, j \in \{1,\dots, \nicefrac{2}{\eta^2}-1\} \,\, \} \cap (0,1].
    \end{equation*}
\end{definition}

\textit{\underline{Preprocessing.}} Next, we will do some preprocessing modifications where each resulting distribution still has at most $O(n \log(1/\tau))$ distribution states, and all distribution states are either in $\G$, or take value $<\tau$.

\begin{definition}[Preprocessing rounding procedure in \cref{alg:preprocess}]
    In the following rounding procedure for a distribution, consider each distribution state one at a time. For some distribution state with value $x$, repeat the following process until $x < \tau$: let $y$ be the largest $y \in \G$ where $y \le x$, then replace $x$ with $y$ and $x-y$, and continue to round $x-y$.
\end{definition}

\begin{claim}\label{claim:preprocessing-rounds}
    After using the preprocessing rounding procedure, each distribution $p_i$ will have at most $n \lceil 1 + \log(1/\tau) \rceil$ states.
\end{claim}
\begin{proof}
    Write $x$ in the form $2^{-i} \eta K$ where $i$ is an integer, and $K$ is not necessarily an integer but it is within $[\frac{1}{\eta^2},\frac{2}{\eta^2})$. Then, $y$ will be at least $2^{-i}\eta \lfloor K \rfloor \in \G$. Moreover,
     \begin{equation}
         x-y \le 2^{-i} \eta = x/K \le \eta^2 x < x/2. \label{eq:rounding-rem}
     \end{equation}
     Hence, the value of $x$ after each round will be less than half its value at the beginning of the round, implying there are at most $\lceil \log(1/\tau) \rceil$ rounds.
\end{proof}

After preprocessing, distribution states are either in $\G$ or are $<\tau$. Later, we will also argue preprocessing does not increase the entropy of the optimal coupling by much. For future notation, let $S_i^{\text{start}}$ be the multiset of values for $p_i$ after this preprocessing.

\textit{\underline{Dynamic program: DP-state space.}} Our DP-state will be specified by: (i) some value $M \in \G \cap [\tau / \alpha,1]$ (for some $0<\alpha < 1$), and (ii) counts (between $0$ and $1/\tau$, inclusive) for each distribution of the states taking values in $[\alpha M, M] \cap \G$. In terms of notation, we will equivalently represent the counts of values by multisets $S_1,\dots, S_m$, where $S_j$ is the multiset of values for $p_j$ within $[\alpha M, M]$. We will use the following claim to help bound the number of potential DP-states.

\begin{claim} \label{claim:bound-G-range}
    For any $x>0$ and $0<\beta \le 1$, it holds that $|\G \cap [\beta x, x]| \le O(\log(\nicefrac{1}{\eta \beta})/\eta^2)$.
\end{claim}
\begin{proof}
    Let $i_{\text{max}}$ denote the largest integer $i$ where $2^{-i} \eta \cdot \frac{2}{\eta^2} \ge \beta x$; observe $i_{\text{max}} \le \log(\nicefrac{2}{\eta \beta x})$. Let $i_{\text{min}}$ denote the smallest integer $i$ where $2^{-i} \eta \le x$; observe $i_{\text{min}} \ge  \log(\eta/x)$. It follows that any $y \in \G \cap [\beta  x,x]$ can be written as some $2^{-i} \eta j$ where $i \in [i_{\text{min}}, i_{\text{max}}]$; implying the number of potential $i$ is bounded by $i_{\text{max}}-i_{\text{min}}+2 \le \log(\nicefrac{1}{\eta^2 \beta}) + 3$. Since there are at most $2/\eta^2$ options for $j$, this gives the bound $(\log(\nicefrac{1}{\eta^2 \beta}) + 3) \cdot (2/\eta^2) = O(\log(\nicefrac{1}{\eta \beta})/\eta^2)$.
\end{proof}

 By invoking \cref{claim:bound-G-range} with $x=1$ and $\beta = \tau$, we can bound the DP-state options for (i) by $O(\log(\nicefrac{1}{\eta \tau})/\eta^2)$. 
 
 Invoking \cref{claim:bound-G-range} with $x=M$ and $\beta = \alpha$, there are at most $O(\log(\nicefrac{1}{\eta \alpha})/\eta^2)$ elements in $[\alpha M, M] \cap \G $. Each distribution's counts have at most $(1 + 1/\tau)^{|[\alpha M, M] \cap \G |}$ options, meaning the DP-state options for (ii) are bounded by $(1 + 1/\tau)^{m|[\alpha M, M] \cap \G |} \le (1/\tau)^{O(m \log(\nicefrac{1}{\eta \alpha})/\eta^2 )}$.

 Combining both bounds, the number of DP-states is at most $(1/\tau)^{O(m \log(\nicefrac{1}{\eta \alpha})/\eta^2 )}$. Later we will set $\tau = \poly(\eta/n)$ and $\alpha = \poly(\eta)$, meaning this is a desirable bound.

\textit{\underline{Dynamic program: first call.}} We will invoke the dynamic program with starting DP-state $M_0=1$ which is the largest value of $\G$, and with each $S_i = [\alpha M_0, M_0] \cap S_i^{\text{start}}$. This begins our dynamic program with the desired invariant, and we will observe that each possible transition will preserve the invariant.

\textit{\underline{Dynamic program: transitions.}} Whenever none of $S_1,\dots,S_m$ contains $M$, we will transition to the DP-state with $M'$, where $M'$ is the largest value of $\G$ that is smaller than $M$. If $M' \ge \tau/\alpha$, then for each set $S_i$ we will add the values of $[\alpha M', \alpha M) \cap S_i^{\text{start}}$. Else, for each set $S_i$ we will add the values of $[\tau, \alpha M) \cap S_i^{\text{start}}$, and the rest is handled later in the base case.

Otherwise, one of the multisets contains $M$, and let $S_{i^*}$ be an arbitrary such multiset. From here, we will consider either splitting this distribution state of size $M$ in half, or creating a coupling state involving it.

For the splitting option, we will simply consider recursing into the dynamic program where $S_{i^*}$ eliminates a copy of $M$, and adds two copies of $M/2$. Note how our set $\G$ was chosen such that for any $M \in \G$, then it holds that $M/2 \in \G$. Thus, our desired invariant holds.

For the coupling option, we will only consider making a state of size $z \in \G \cap [\eta M, M]$. By \cref{claim:bound-G-range}, there are at most $O(\log(\nicefrac{1}{\eta})/\eta^2)$ options for $z$. Then, we will choose the size of distribution state to couple from each distribution. Let $x_1,\dots,x_m$ denote the size chosen for each distribution. Each $x_i$ must satisfy $x_i \in S_i$ and $x_i \ge z$; additionally $x_{i^*}=M$. By \cref{claim:bound-G-range}, there are $O(\log(\nicefrac{1}{\eta})/\eta^2)^{m-1}$ options for the vector $x$; hence our DP has a total of $O(\log(\nicefrac{1}{\eta})/\eta^2)^m$ possible transitions. Still, after having chosen $z$ and $x_1,\dots,x_m$, we must design a way to create this coupling state and recursively call the DP in a way that maintains our invariant. For example, naively modifying each distribution state $x_i$ to be $x_i - z$, would be invalid because often $x_i - z \notin \G$. Our goal will be to split every $x_i-z$ into two terms $a_i + b_i$, where each $a_i,b_i$ are in $\G$, and are either $0$ or at least $\eta^3M/2$. The last property is crucial, because if we choose $\alpha \le \eta^3/2$ then this implies our invariant still holds. Hence, we will choose $\alpha = \eta^3 / 2$.

Let us define a function $\match: \G^{m+1} \rightarrow (\G \cup \{0\})^{2m}$ where $\match(z,x_1,\dots,x_m)$ outputs $ a_1,\dots, a_m, b_1,\dots,b_m$. We will show such a function exists with our desired properties:

\begin{claim}\label{claim:ab-exist}
    There exists a function $\match: \G^{m+1} \rightarrow (\G \cup \{0\})^{2m}$, such that for any $z,x_1,\dots,x_m \in \G$ where $M = \max_i x_i$ and $\eta M \le z \le \min_i x_i$, then it holds that $\match(z,x_1,\dots,x_m)$ outputs  $ a_1,\dots, a_m, b_1,\dots,b_m$, where all (i) $a_i,b_i \in \G \cup \{0\}$, (ii) $x_i = z + a_i + b_i$, (iii) $a_i,b_i$ are either zero or at least $\eta^3 M/2$, and (iv) $a_i \le \eta z$.
\end{claim}
\begin{proof}
    Note that property (iv) is not immediately relevant, but it will be used in our approximation error analysis.
    
    We will first aim to find valid $a_k,b_k$ with the correct sum. For any value $x_k \in \G$, it may be written as $2^{-i} \eta K$ for integers $i,K$ where $1 \le K \le \nicefrac{2}{\eta^2}-1$, but there are sometimes multiple combinations of $i,K$ that yield value $x_k$. We define the \textit{maximal form} of $x_k \in \G$ as the pair where $i$ is maximized. Note how in maximal form, the coefficient $K \ge 1/\eta^2$. Write the maximal forms of $z,x_k$ as $z = 2^{-i} \eta K$ and $x_k = 2^{-j} \eta L$. If $i=j$, then $x_k-z \in \G \cup \{0\}$, so set $a_k=0$ and $b_k = x_k-z$.

    Otherwise, $i > j$. We will show it is possible to write $x_k-z$ in the form $2^{-i} \eta A + 2^{-j} \eta B$, where $a_k = 2^{-i} \eta A$ and $b_k = 2^{-j} \eta B$. 

    Choose $B$ to be the largest integer where $2^{-j} \eta B \le x_k-z$. Clearly $0 \le B < L$, so $b_k$ is either $0$ or in $\G$.
    
    The remainder, $x_k-z-b_k$, is certainly a non-negative integer multiple of $2^{-i} \eta$, and moreover the remainder is strictly less than $2^{-j} \eta$. We will conclude the remainder is at most $\eta z$, which would imply both  $a_k \in \G \cup \{0\}$ (since then $0 \le A \le K$), and condition (iv) that $a_k \le \eta z$:

    \begin{align*}
        & a_k = x_k-z-b_k < 2^{-j} \eta = \frac{x_k}{L} \intertext{since $x_k$ is in maximal form, $L \ge \frac{1}{\eta^2}$:}
        & \le \eta^2 x_k \le \eta^2 M \le \eta z.
    \end{align*}

    All that remains is condition (iii). Observe how all of $a_k,b_k$ were either $0$ or at least $2^{-i} \eta$. This implies condition (iii) since
    \begin{equation*}
        2^{-i} \eta \ge \eta^2 z/2 \ge \eta^3 M /2.
    \end{equation*}
    Finally, note how computing $\operatorname{Match}(z,x_1,\dots,x_m)$ runs in polynomial time, as computing each $a_k,b_k$ is just constant-time casework given the maximal forms of $z$ and $x_k$, and it is simple to compute these maximal forms (e.g. you could even afford to naively try all the $\nicefrac{2}{\eta^2}-1$ possible values for $K$). 
\end{proof}

Leveraging the guarantees of \cref{claim:ab-exist}, for the chosen  $z,x_1,\dots,x_m$ we will create a coupling state of size $z$, costing $\phi(z)$ entropy, and then will use $\match(z,x_1,\dots,x_m)$ to recursively call our dynamic program with modified sets $S_i' \triangleq S_i \backslash \{x_i\} \cup \{a_i \} \cup \{b_i \}$. Since we will choose $\alpha = \eta^3/2$, and each nonzero $a_i,b_i \ge \eta^3 M /2$, our invariant holds.

Our dynamic program will choose the minimum-entropy cost of all considered options.

\textit{\underline{Dynamic program: base case.}} Our base case is whenever $M < \tau/\alpha$. In this case, let us define the set of leftover distribution states as $S_i^{\text{leftover}} \triangleq S_i \cup (S_i^{\text{start}} \cap (0, \tau ))$. It is known that the \textit{maximum-entropy coupling} is simply the coupling where the distributions are independent. This means no matter how the leftover distribution states $S_1^{\text{leftover}},\dots,S_m^{\text{leftover}}$ are coupled, the entropy will be at most $\sum_i H(S_i^{\text{leftover}})$. We will choose to have our dynamic program return this upper bound $\sum_i H(S_i^{\text{leftover}})$. In a post-processing of our dynamic program, we will use some well-known greedy coupling method (e.g. \cite{kocaoglu2017entropic}) to handle the leftover distribution states (and this will certainly have entropy at most that of the maximum-entropy coupling upper bound the dynamic program returned). 

\textit{\underline{Constructing the coupling.}} We will sketch how to construct a coupling (beyond just knowing the entropy value) with entropy upper bounded by the cost returned by the dynamic program. It is a standard technique to ``trace through'' a dynamic program to construct the corresponding solution. In typical fashion, start at the first call of the dynamic program. Then, whichever is the minimum-cost transition (either splitting the largest distribution state, or making a coupling state and adjusting the multisets), take the same corresponding action in the coupling construction. Eventually, when reaching the base case, use the greedy coupling of \cite{kocaoglu2017entropic} to couple the leftover distribution states. As argued in the previous paragraph, this greedy coupling must have entropy bounded by the base case value returned by the dynamic program. The details of the greedy coupling algorithm are inconsequential, all we leverage is that it produces a valid coupling and has bounded running time (for this one, it is $O(m^2 n \log(n))$ time); there are multiple alternative coupling algorithms that would also satisfy these properties. For any base case, the multisets $S_1^{\text{leftover}},\dots,S_m^{\text{leftover}}$ will each have size at most $1/\tau + 2n$, since all values will be at least $\tau$, other than at most $ 2n$ leftover distribution states of size $<\tau$ from the rounding preprocessing step; this means the greedy coupling runs in time $O(m^2(n + 1/\tau) \log(n + 1/\tau))$. Note also how our algorithm may have split distribution states, and the corresponding coupling states should be mapped back to the original distribution state from the input; this may cause some coupling states to be merged, but this would only decrease entropy.

\textit{\underline{Running time.}} The total running time is $(1/\tau)^{O(m \log(\nicefrac{1}{\eta \alpha})/\eta^2 )}$, which is dominated by the number of DP-states. A pedantic point is that technically it is not known how to exactly compute the simple expression $\phi(x) = x \log(1/x)$ in polynomial time. See \cref{appendix:precise} for some further discussion, where we explain how it is sufficient to approximately compute these terms. For simplicity, we will write the algorithm as if the entropy terms are being computed exactly, but nothing fundamentally changes when considering the true version (the runtime bound remains the same).

Pseudocode is provided in \cref{alg:preprocess,alg:fast}. For context, recall that we will choose $\alpha = \eta^3 / 2$.

\begin{algorithm}[]
    \caption{Preprocessing for fast algorithm}
   \label{alg:preprocess}
\begin{algorithmic}[1]
    \State {\bfseries Input:} Discrete probability distributions $p_1,\dots,p_m$ each with at most $n$ states.
    \Procedure{Rounding}{$x$}
    \If{$x < \tau$}
    \Return $\{x\}$
    \EndIf
    \State $y \gets $ largest value of $\G$ that is not greater than $x$
    \State \Return $\{y \} \cup  \operatorname{Rounding}(x-y)$ \Comment{This is a multiset}
    \EndProcedure
    \Procedure{Processing}{$p_1,\dots,p_m$}
    \For{$i \in \{1,\dots,m\}$}
    \State $S_i^{\text{start}} \gets \emptyset$ \Comment{Multiset of values for $p_i$}
    \For{$j \in \{1,\dots,n\}$}
    $S_i^{\text{start}} \gets S_i^{\text{start}} \cup \operatorname{Rounding}(p_i(j))$
    \EndFor
    \EndFor
    \State $M_0 \gets 1$ 
    \For{$i \in \{1,\dots,m\}$}
    $S_i \gets S_i^{\text{start}} \cap [\alpha M_0,M_0]$
    \EndFor
    \State couplingCost $\gets \operatorname{FastCoupling(M_0,S_1,\dots,S_m)}$
    \UnindentedComment{Construct a coupling with cost at most couplingCost by following the steps in \textit{``Constructing the coupling''}}
    \EndProcedure
\end{algorithmic}
\end{algorithm}

\begin{algorithm}[]
    \caption{Dynamic programming algorithm}
   \label{alg:fast}
\begin{algorithmic}[1]
    \State {\bfseries Input:} $M \in \G$. Multisets $S_1,\dots,S_m$ of elements of $\G$ detailing the remaining states of $p_1,\dots,p_m$ in $\G \cap [\alpha M, M]$.
    \Procedure{FastCoupling}{$M,S_1,\dots,S_m$} 
    \If{DP[$M,S_1,\dots,S_m$] already computed}
        \Return DP[$M,S_1,\dots,S_m$]
    \EndIf
    \If{$M <  \tau / \alpha$}
        \For{$i \in \{1,\dots m\}$} 
        $S_i^{\text{leftover}} \gets S_i \cup (S_i^{\text{start}} \cap (0,\tau))$
        \EndFor
        \Return $\sum_{i=1}^m H(S_i^{\text{leftover}})$
    \EndIf
    \If{$M \notin (S_1 \cup \dots \cup S_m)$}
        \State $M' \gets $ largest element of $\G$ less than $M$
        \For{$i \in \{1,\dots m\}$} 
        $S_i' \gets S_i \cup (S_i^{\text{start}} \cap [\max(\tau,\alpha M'),\alpha M))$
        \EndFor
        \State DP[$M,S_1,\dots,S_m$] $ \gets $ $\operatorname{FastCoupling}(M',S_1',\dots,S_m')$
        \State \Return DP[$M,S_1,\dots,S_m$]
    \EndIf
    \State $i^* \gets $ arbitrary value satisfying $M \in S_{i^*}$
    \UnindentedComment{Consider splitting the $M$-size state of $p_{i^*}$ in half}
    \State DP[$M,S_1,\dots,S_m$] $\gets \operatorname{FastCoupling}(M,S_1,\dots,S_{i^*} \backslash \{M\} \cup \{M/2\} \cup \{M/2\},\dots,S_m)$
    \For{$z \in \G \cap [\eta M,M]$, $x_1,\dots,x_m$ where $x_i \in S_i \cap [z, M]$ and $x_{i^*}=M$} 
    \UnindentedComment{Consider coupling $z$ mass, where the coupling involves a state of size $x_i$ from distribution $p_i$}
    \State $a_1,\dots,a_m,b_1,\dots,b_m \gets \match(z,x_1,\dots,x_m)$
    \For{$i \in \{1,\dots m\}$} 
    $S_i' \gets S_i \backslash \{x_i\} \cup \{a_i\} \cup \{b_i\}$ \Comment{Ignore any $\{0\}$ added to the multiset}
    \EndFor
    \State DP[$M,S_1,\dots,S_m$] $\gets \min(\text{DP}[M,S_1,\dots,S_m], \phi(z) + \operatorname{FastCoupling}(M,S_1',\dots,S_m')) $
    \EndFor
    \State \Return DP[$M,S_1,\dots,S_m$].
    \EndProcedure
\end{algorithmic}
\end{algorithm}

\section{Algorithm approximation error analysis} \label{sec:analysis} 
In the previous section we have discussed an efficient algorithm which produces valid couplings, but we have not justified its approximation error. We will show the approximation error has some bound in terms of $\eta$ and $m$, and then eventually we will invoke \cref{alg:fast} with a smaller $\eta$ defined in terms of $\eps$ and $m$ to conclude an additive approximation error of $\eps$.

Our general proof technique will start with some coupling $\sol$ where $H(\sol)=H(\opt)$, and iteratively modify it in a way such that it can be consistent with a path in the dynamic program of \cref{alg:fast}. As we modify the solution, we will call this modified coupling $\sol'$, and we will show the final $H(\sol')$ is not much larger than $H(\sol)$. This will upper bound the value returned by the dynamic program, which in turn bounds the entropy of our coupling.

The process of modifying $\sol$, so that it may be found by the dynamic program, is itself quite algorithmic. There are four categories of actions taken in the process of \cref{alg:preprocess,alg:fast} after which we must adjust $\sol$:

\begin{enumerate}
    \item \textit{Preprocessing.} During the preprocessing phase, input distribution states $p_i(j)$ are rounded to values in $\G$ (and one value not in $\G$ but $<\tau$), and in the process are split into multiple distribution states. We must split coupling states of $\sol'$ so they are consistent with this new split version of distribution states. Our main intuition for bounding the increase of $H(\sol')$ is that because the sequence of rounded values for the new distribution states is bounded by a geometric series with large ratio, this will not influence entropy much.
    \item \textit{Matching.} Whenever we choose a matching of size $z$ in the dynamic program, this will correspond to our attempt to simulate a coupling state of size $q \in \sol'$ where $q-z \le 2 \eta z$. Still, this is not exactly creating a matching of size $q$, and it is also splitting up distribution states into $x_i-z=a_i+b_i$, so we must modify $\sol'$ to be consistent. Our main intuition for bounding the increase of $H(\sol')$ is that because $q \approx z$, this is close to simulating a coupling state in $\sol'$.
    \item \textit{Splitting.} When splitting a distribution state of size $M$ in half, this causes a need to modify $\sol'$ as well. This is the most difficult source of approximation error we need to bound, and the charging method is nuanced. Very roughly, our charging argument will maintain a sorting of the coupling states of $\sol'$ inside each distribution state of $S_1,\dots,S_m$ (with an atypical sorting order), and it will be beneficial to split the state of $\sol'$ at the midpoint of this sorting. 
    \item \textit{Leftover.} In the base case, we upper bound the error in the remaining coupling by $\sum_{i=1}^m H(S_i^{\text{leftover}})$. We need to show that this quantity is not much larger than the entropy of the remaining uncoupled portions of $\sol'$. Our main intuition for this is to directly bound the quantity $\sum_{i=1}^m H(S_i^{\text{leftover}})$, leveraging some structure in our $\sol'$ modification procedure that implies there are not too many leftover states in $\sol'$. If we could only bound that there were at most, say, $1/\tau$ states with size at most $\tau$ in the end, then the entropy could be quite large; we will show the number of states at the end (of our particular solution path in the DP) is much smaller.
\end{enumerate}

We will now discuss how we modify $\sol$ for each type of action, and bound the resulting approximation error. Throughout this section, it is informative to view couplings pictorially; see \cref{fig:coupling,fig:split-coupling}.

\begin{figure}[H]
    \centering
    \includegraphics[width=\textwidth]{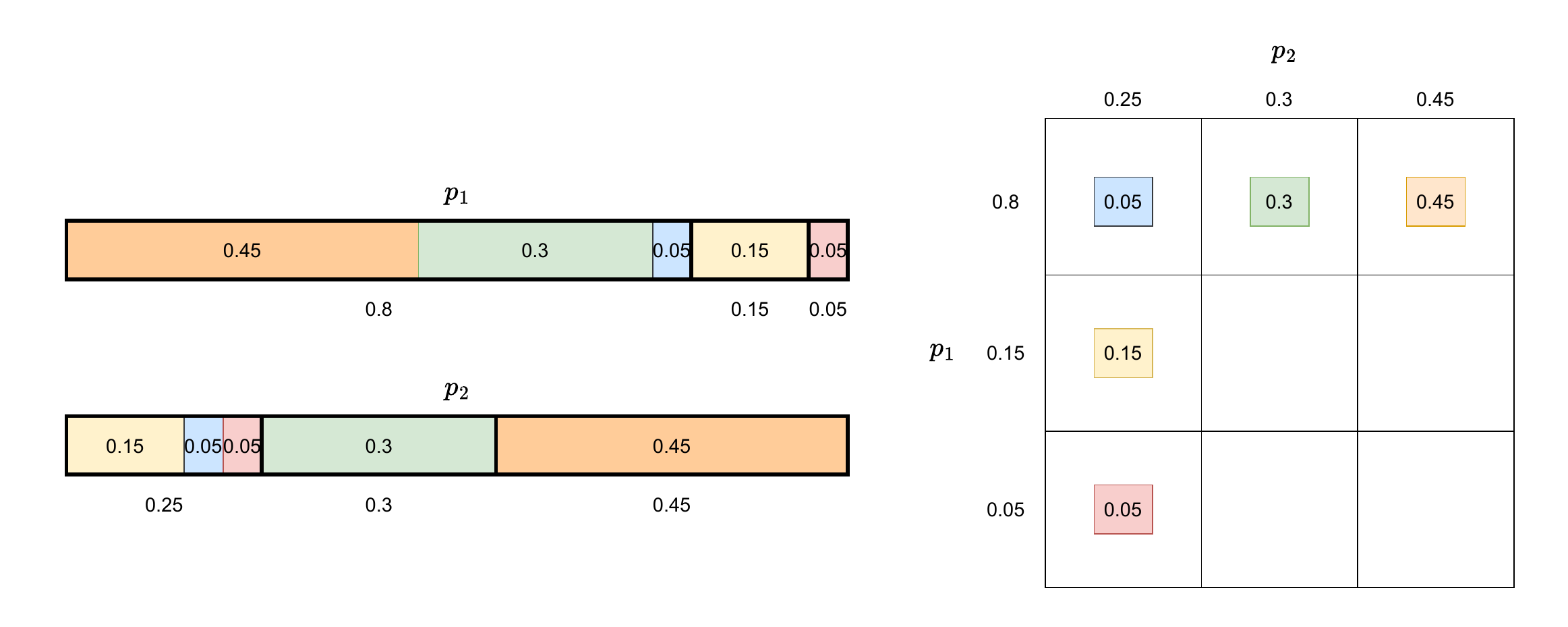}
    \caption{Coupling between two distributions $p_1 = [0.8,0.15,0.05]$ and $p_2 = [0.25,0.3,0.45]$. The left half of the figure has a rectangular depiction of the distributions, where the bold borders represent the states of the distributions, and the subdivisions by colors represent states of the coupling. The right half of the figure gives the analogous table view of the coupling.}
    \label{fig:coupling}
\end{figure}

\begin{figure}[H]
    \centering
    \includegraphics[width=\textwidth]{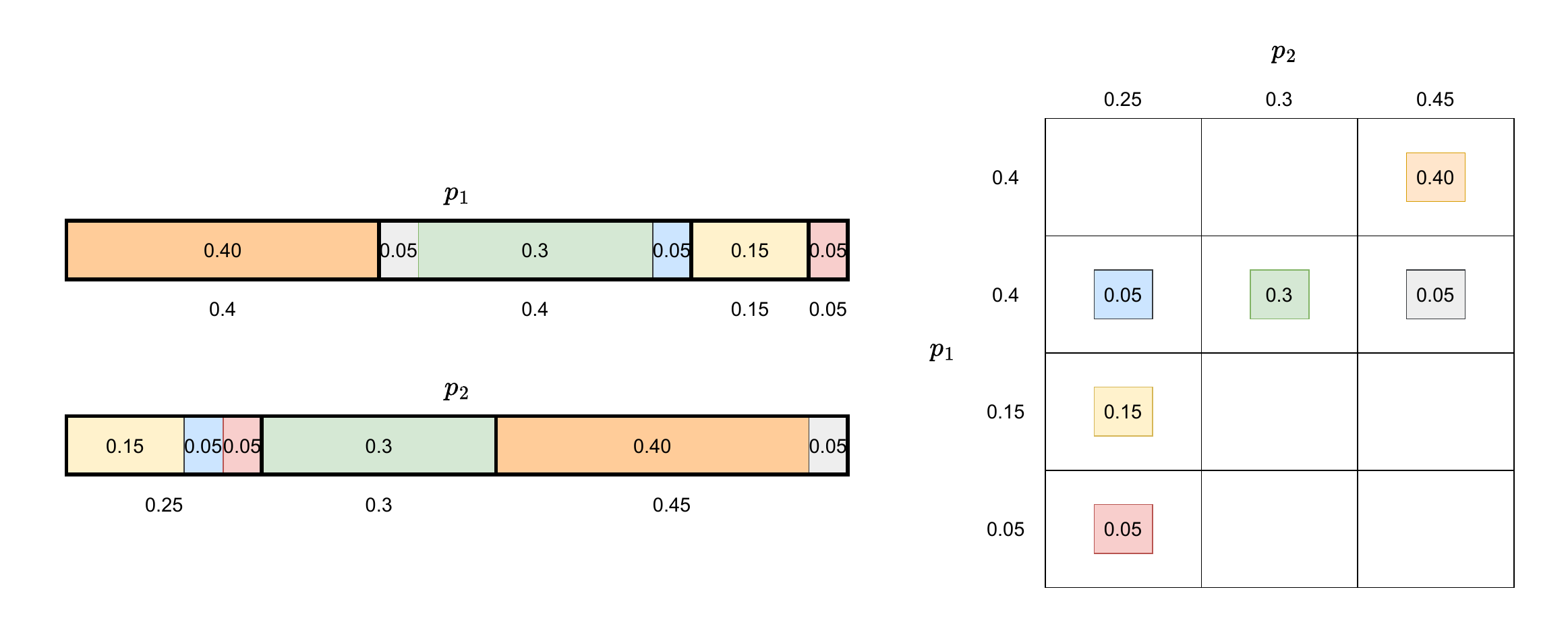}
    \caption{Depiction of the coupling in \cref{fig:coupling} after the first state is split in half. The $p_1$ state of size $0.8$ is split into two states of size $0.4$, thus also splitting the coupling state of size 0.45 (orange) into coupling states of size $0.4$ (orange) and $0.05$ (gray).}
    \label{fig:split-coupling}
\end{figure}

\subsection{Preprocessing}

\textbf{Modification procedure.} We start with any optimal coupling $\sol$ for the input distributions. In the preprocessing stage of \cref{alg:preprocess}, our modifications to $\sol'$ will be straightforward. Any time a distribution state of size $x$ is split into two distribution states of size $x-y$ and $y$, we will simply split $\sol'$ at the border, as depicted in \cref{fig:split-coupling}.

\textbf{Entropy increase analysis.} We start by stating some simple bounds on the increase in entropy when splitting a state:
\begin{claim}\label{claim:simple1}
    Consider splitting a probability mass $x$ into two states $cx$ and $(1-c)x$ for $c \le \nicefrac{1}{2}$. The entropy is bounded by
    \begin{equation*}
        \phi((1-c)x) + \phi(cx) \le \phi(x) + cx/\ln(2) + cx \log(1/c).
    \end{equation*}
\end{claim}
\begin{proof}
    \begin{align*}
        & \phi((1-c)x) + \phi(cx) = \phi(x) + (1-c)x \log\left(\frac{1}{1-c}\right) + cx \log\left( \frac{1}{c} \right) \\
        & \le \phi(x) + (1-c) x \cdot \frac{c}{(1-c)\ln(2)} +cx \log(1/c) = \phi(x) + cx/\ln(2) + cx \log(1/c) \quad\qedhere
    \end{align*}
\end{proof}

\begin{corollary}\label{cor:split-c}
    Consider splitting a probability mass $x$ into two states, where the smaller state is at most $cx$ for $c \le 1/e$. Then, the entropy is bounded by 
    \begin{equation*}
        \phi((1-c)x) + \phi(cx) \le \phi(x) + cx/\ln(2) + cx \log(1/c).
    \end{equation*}
\end{corollary}
\begin{proof}
    This follows from how the bound in \cref{claim:simple1} is non-decreasing for $c \le 1/e$.
\end{proof}

\begin{corollary} \label{cor:split-y-bound}
    Consider splitting a probability mass $x$ into two states, where one state is $y$ and the other state is $x-y$. Then, the entropy is bounded by
    \begin{equation*}
        \phi(x-y) + \phi(y) \le \phi(x) + ey + y \log(x/y).
    \end{equation*}
\end{corollary}
\begin{proof}
    \textbf{Case 1: $y \le x/e$. } By \cref{cor:split-c}, we have a bound of
    \begin{equation*}
        \phi(x-y) + \phi(y) \le \phi(x) + y/\ln(2) + y \log(x/y) 
    \end{equation*}

    \textbf{Case 2: $y > x/e$. } By concavity of entropy
    \begin{equation*}
        \phi(x-y) + \phi(y) \le 2\phi(x/2) = \phi(x) + x \le \phi(x) + ey \qedhere
    \end{equation*}
\end{proof}

These bounds on entropy increase from splitting will be enough to prove an upper bound on the increase in entropy of $\sol'$ throughout preprocessing:

\begin{lemma}[Rounding initial input to $\G$]\label{lemma:rounding-error}
    After rounding a distribution with the preprocessing procedure in \cref{alg:preprocess}, the entropy of $\sol'$ will increase by at most $2 \eta$.
\end{lemma}
\begin{proof}
     In this lemma, we focus on some input distribution $p_i$. For some distribution state with initial value $x \triangleq p_i(j)$, let us track the recursive rounding process with some vector, where $x_0=x$, then $x_1$ is the value of $x-y$ in the first round, and so on. By  \cref{eq:rounding-rem}, we know $x_{i+1} \le \eta^2 x_i$. Let $\Delta_i$ denote the entropy increase to $\sol'$ from the split at phase $i$. Suppose phase $i$ caused some coupling state of $\sol'$ with mass $c_i$ to be split into $d_i$ and $c_i-d_i$ (where we choose $d_i < c_i-d_i$). Then using \cref{cor:split-y-bound}:
     \begin{align*}
         \Delta_i &\triangleq \phi(c_i-d_i) + \phi(d_i) - \phi(c_i) \le e d_i + d_i \log(c_i / d_i) \le e d_i + d_i \log(x_{i-1} / d_i) \intertext{Since $d_i \log(x_{i-1}/d_i)$ is non-decreasing in $d_i$ for $[0,x_{i-1}/e]$, and it must hold that $d_i \le x_i \le \eta^2 x_{i-1} \le x_{i-1}/e$:}
         & \le e x_i + \eta^2 x_{i-1} \log(1/\eta^2) \le \eta^2 x_{i-1} \cdot (e + \log(1/\eta^2)) 
     \end{align*}
     We may now bound the total increase in entropy from rounding one state:
     \begin{align*}
         \sum_{i \ge 1} \Delta_i &\le \sum_{i \ge 1} \eta^2 x_{i-1} \cdot (e + \log(1/\eta^2)) \\
         & \le \sum_{i \ge 1} \eta^{2i} x_0 \cdot (e + \log(1/\eta^2)) \intertext{Since the ratio of consecutive terms is less than $1/2$:}
         & \le 2 \eta^2 x_0 \cdot (e + \log(1/\eta^2)) 
     \end{align*}
     Note how the sum of the values of $x_0$ for each state of $p_i$ is $1$. This implies that after all states are rounded for some distribution, the solution entropy will increase by at most $2 \eta^2 \cdot (e + \log(1/\eta^2))$. For $\eta < \nicefrac{1}{10}$, this is at most $2 \eta$.
\end{proof}

An immediate corollary of \cref{lemma:rounding-error} is that after all $m$ distributions are preprocessed, the increase in entropy to $\sol'$ is at most $2 \eta m$.

\subsection{Setup for dynamic programming related modifications}
Our modifications occurring inside the dynamic program will be more nuanced. Observe how in \cref{fig:coupling}, the first state of $p_1$ contains three coupling states of $\sol$ (of size 0.45, 0.3, and 0.05). Within this state of $p_1$, we could order the three states of the coupling however we like. We will use a special ordering procedure for coupling states within a probability distribution state. Informing this ordering, we define four types of coupling states: \textit{regular}, \textit{important}, \textit{left-pierced}, and \textit{right-pierced}. The ordering will be the left-pierced coupling states (in an arbitrary order), followed by the regular coupling states sorted in non-increasing order, followed by the important coupling states (in an arbitrary order), followed by the right-pierced coupling states (in an arbitrary order). At the start of the dynamic program, all coupling states are considered regular. Note that the type of each coupling state is determined separately for each distribution; for example, $p_1$ may consider a coupling state to be regular, while $p_2$ considers the state important, and $p_3$ considers the state left-pierced.

We now detail which choice in the dynamic program we will choose, depending on $\sol'$. If $M< \tau/\alpha$, we will handle the base case with leftovers. Recall $i^*$ is an arbitrary value such that $M \in S_{i^*}$. If there is no such $i^*$, then our dynamic program will keep reducing $M$ until it is the size of the largest remaining distribution state (this does not affect $\sol'$). Once $i^*$ exists, let $\ind_{i^*}$ be our reference for an arbitrary remaining distribution state of $p_{i^*}$ with size $M$. Consider two cases: (i) $\ind_{i^*}$ has a coupling state in $\sol'$ of size at least $2\eta M$, or (ii) there is no such sufficiently large coupling state in $\ind_{i^*}$. In the first case, we will do a matching; in the second case, we will split $\ind_{i^*}$ in half. Let us more precisely describe how we modify $\sol'$ in these different actions.

\subsection{Matching}
\textbf{Modification procedure.} Let $q$ be the size of the coupling state in $\ind_{i^*}$ where $2 \eta M \le q \le M$ (choose an arbitrary $q$ if there are multiple). Since it may be the case that $q \notin \G$, our dynamic program may not be able to exactly choose a coupling state of size $q$. We will show the existence of a $z \in \G$ where $z \approx q$. While $z$ will be close to $q$, we will also desire to show some buffer (i.e. they are not extremely close to each other) so that we may do required adjustments from handling the extra splitting from $a_1,\dots,a_m,b_1,\dots,b_m$ (this will be more clear later). We show:
\begin{claim}\label{claim:z-exists}
    For any value $0 < q \le 1$, there exists a $z \in \G$ where $1.5\eta q \le q-z \le 2 \eta z$.
\end{claim}
\begin{proof}
    Write $q$ in a form $2^{-i} \eta L$ where $i$ is an integer and $L \in [\nicefrac{1}{\eta^2},\nicefrac{2}{\eta^2})$ is not necessarily an integer. By definition of $\G$, such a form must exist. 

    We will choose the value $z = 2^{-i} \eta \cdot \lfloor L - 1.5 \eta L \rfloor$. This choice of $z$ is a valid element in $\G$ because the coefficient $\lfloor L - 1.5 \eta L \rfloor$ is an integer less than $L$, and it is a positive integer since
    \begin{equation*}
        \lfloor L - 1.5 \eta L \rfloor \ge L \cdot (1-1.5\eta) -1 \ge \frac{1}{2}L-1 \ge \frac{1}{2 \eta^2} - 1 \ge 49.
    \end{equation*}

    Moreover, we observe $q-z$ satisfies our desired properties:

    \begin{equation*}
        \frac{q-z}{z} = \frac{L - \lfloor L - 1.5 \eta L \rfloor}{\lfloor L - 1.5 \eta L \rfloor} \le \frac{L - (L - 1.5\eta L -1)}{L - 1.5\eta L - 1} = \frac{1.5 \eta L + 1}{L -1.5\eta L -1} \le \frac{1.6 \eta L}{L-1.6\eta L } \le \frac{1.6 \eta }{1-0.16} \le 2 \eta \implies q-z \le 2 \eta z
    \end{equation*}

    \begin{equation*}
        q-z = \frac{L - \lfloor L - 1.5\eta L \rfloor }{L} \cdot q \ge \frac{1.5 \eta L}{L} \cdot q = 1.5 \eta q \qedhere
    \end{equation*}
\end{proof}

In the dynamic program, we will choose this value of $z$; observe that $z$ is an eligible option because \cref{claim:z-exists} and $\eta \le \frac{1}{2}$ imply $z \ge q/2 \implies z \ge \eta M$. Let $\ind_1,\dots,\ind_m$ represent the states of $p_1,\dots,p_m$ that the coupling state (of size $q$) maps to. We will choose the values of $x_1,\dots,x_m$ corresponding to the sizes of $\ind_1,\dots,\ind_m$. Then, $\match(z,x_1,\dots,x_m)=a_1,\dots,a_m,b_1,\dots,b_m$ satisfies properties given by \cref{claim:ab-exist}. We must modify $\sol'$ so that it is still a valid coupling after every distribution state $x_i$ is split into $z,a_i,b_i$. Observe by \cref{claim:z-exists} and property (iv) of \cref{claim:ab-exist}, that all $a_i \le \eta z < 1.5\eta q \le q-z$. 

This gives us a clear modification plan. We will split the coupling state of size $q$ at the points $z$ and $z+a_i$ for all $i$. By the previous statement, all $z+a_i \in [z,q)$. Observe how this modified $\sol'$ is now a valid coupling after each distribution state was split into $z,a_i,b_i$: the coupling state of size $z$ will map to the new distribution state of size $z$, all coupling states ending by $z+a_i$ will map to the new distribution state of size $a_i$, and the remaining coupling states will be left with the distribution state of size $b_i$. In total, the $q-z$ remaining coupling mass was split into at most $m+1$ new coupling states. In terms of bookkeeping, we will now ignore the new coupling state of size $z$ since it is fully coupled, all coupling states in the new distribution state of size $a_i$ will be considered regular, in the distribution state $b_i$ we will consider its new coupling states as regular (these are the coupling states that came from the range $[z+a_i,q]$), but the existing coupling states in $b_i$ will retain their previous type. This is depicted in \cref{fig:match}. Informally, it may be helpful to think of this process as modifying a distribution state $\ind_i$ by ``deleting'' $z$ mass, forking off $a_i$ mass into a new distribution state, and $b_i$ mass is what remains of the initial state $\ind_i$. 

\begin{figure}[H]
    \centering
    \includegraphics[width=\textwidth]{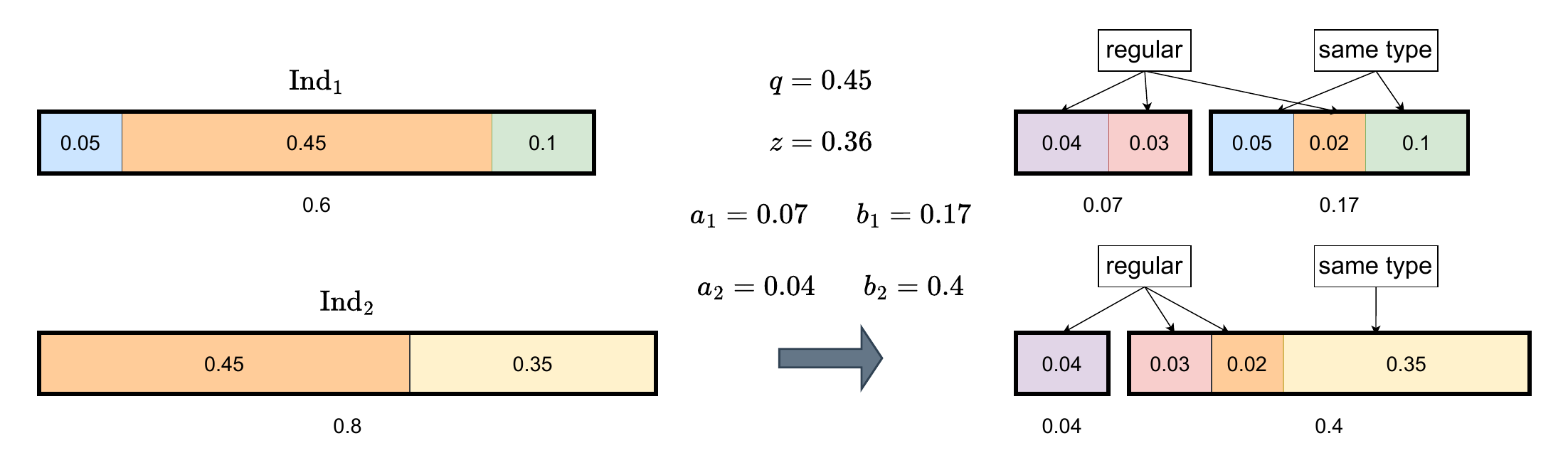}
    \caption{Depiction of how $\sol'$ is modified during a matching modification. This figure looks at how distribution states $\ind_1$ of $p_1$, and $\ind_2$ of $p_2$, are modified when a coupling state of size $q=0.45$ motivates the dynamic program to match with $z=0.36$. $\ind_1$ is split into states of size $z=0.36$, $a_1=0.07$, and $b_1=0.17$ (the $z$ state is ignored since it is now fully coupled). $\ind_2$ is split into states of size $z=0.36$, $a_2=0.04$, and $b_2=0.4$ (the $z$ state is similarly ignored). The coupling state of size $q$ is accordingly split at points $z,z+a_1,z+a_2$, and the three new coupling states are considered regular while the existing coupling states keep their original type. Following this, coupling states are reordered according to their types and sizes.} 
    \label{fig:match}
\end{figure}

\textbf{Entropy increase analysis.} We show how each matching modification does not increase the entropy of $\sol'$ by much:

\begin{lemma}[Matching modification bound]
    Any matching modification does not increase the entropy of $\sol'$ by more than
    \begin{equation*}
        z \cdot \eta \cdot (4 \log(1/\eta) + 2 \log(m+1)).
    \end{equation*}
\end{lemma}
\begin{proof}
    Recall how the only modification to $\sol'$ was how a coupling state of size $q$ was split into a coupling state of size $z$, in addition to at most $m+1$ more coupling states with total mass $q-z$. First, consider the entropy increase from splitting the coupling state into $z$ and $q-z$. By \cref{cor:split-y-bound} and $1.5 \eta q \le q-z \le 2 \eta z$, this increases entropy by at most 
    \begin{equation*}
        e(q-z) + (q-z) \log(q/(q-z)) \le 2e \eta z + 2\eta z \log(1/\eta) \le 4 \eta z \log(1/\eta).
    \end{equation*}
    Later, splitting the $q-z$ state into at most $m+1$ states cannot increase the entropy by more than $(q-z) \log(m+1) \le 2 \eta z \log(m+1)$. In total, the increase in entropy is at most $z \cdot \eta \cdot (4 \log(1/\eta) + 2 \log(m+1))$.
\end{proof}

Since the sum of $z$ matched throughout the entire process is at most $1$, these matching operations increase entropy by at most $\eta \cdot (4 \log(1/\eta) + 2 \log(m+1)) \le 6 \eta \log(1/\eta)$, since $1/\eta \ge 2^{3m} \cdot 4m > m+1$.

\subsection{Splitting} 
\textbf{Modification procedure.} When a distribution state $\ind_{i^*}$ is split in half, $\sol'$ is simply split as is done in \cref{fig:split-coupling}. However, we do make modifications to the types of various coupling states, which affects the ordering after the split. When the distribution state $\ind_{i^*}$ is split in half, let us denote $s_L$ as the left half, and $s_R$ as the right half. None of the types in $s_R$ are modified; if the leftmost coupling state of $s_R$ is a new coupling state created by the split, then we consider it left-pierced. For $s_L$, we will consider its rightmost coupling state as right-pierced, consider its next $\nicefrac{1}{8\eta}$ rightmost coupling states as important, and the remaining coupling state types are unchanged (later it will be clear that they were all either left-pierced or regular). 
 
So far, we have only detailed how the splitting operation changes the coupling state types with respect to $p_{i^*}$, but the types for other distributions are also affected. Let $\mathcal{C}$ refer to the coupling state that was split into two parts by this operation (if there is no such $\mathcal{C}$, because the split occurred between two coupling states, then no further type changes are required). For all $j \ne i^*$: if $p_j$ considered $\mathcal{C}$ (left/right)-pierced or regular, then both parts retain the original type; else, $p_j$ previously considered $\mathcal{C}$ important, and it will now consider the smaller split component of $\mathcal{C}$ to be regular, while the larger component retains the important type (if equally sized, tiebreak arbitrarily). The type changes are visualized in \cref{fig:split-types}.

\begin{figure}[H]
    \centering
    \includegraphics[width=\textwidth]{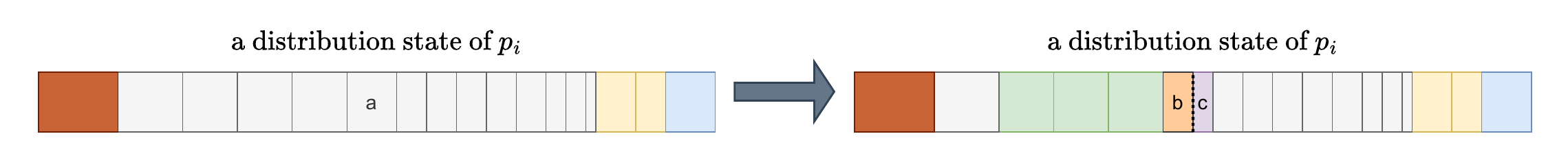}
    \caption{Depiction of how splitting some distribution state of $p_i$ in half affects its types of coupling states. Initially, $p_i$ considered one coupling state in this distribution state as left-pierced (brown), two as important (yellow), and one as right-pierced (blue). Coupling state $a$ is regular (gray) and contains the midpoint. After the split operation, the existing coupling states to the right of the midpoint retain their original type. Coupling state $c$ becomes left-pierced, and coupling state $b$ becomes right-pierced. The next $\nicefrac{1}{8 \eta}$ coupling states to the left become important (green). The remaining coupling states to the left will retain their original type. For any other distribution $p_j$, this may only affect the types of $b$ and $c$. If $p_j$ considered $a$ (left/right)-pierced or regular, then $b$ and $c$ will retain this type; else, $p_j$ previously considered $a$ important, and it will now consider $b$ important (because it is larger), and $c$ regular.} 
    \label{fig:split-types}
\end{figure}

Afterwards, reorder the coupling states for all distributions according to their type and size.

\textbf{Entropy increase analysis.} As we analyze, it will be helpful to have specific terminology to discuss how coupling states are modified over time. We will define a special tree that characterizes these modifications; please refer to \cref{fig:tree} during the following paragraph while we establish the relevant definitions for this tree.

\begin{figure}[H]
    \centering
    \includegraphics[width=0.6\textwidth]{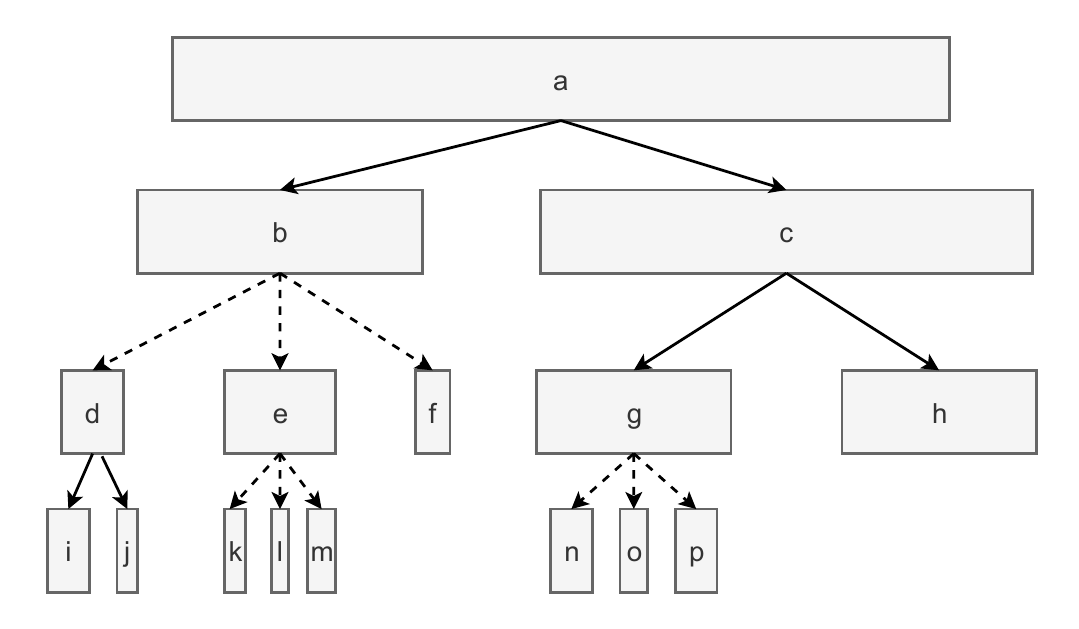}
    \caption{Depiction of how a coupling state is modified over time. Please see the following paragraph for the relevant definitions.} 
    \label{fig:tree}
\end{figure}

Nodes in the tree represent a coupling state, and outgoing edges represent when a coupling state is modified, resulting in its children. The root is a coupling state in $\sol'$ at the start of the dynamic program. Solid outgoing arrows represent a modification caused by a splitting operation (the sum of the two child sizes will be the same as the parent size), while dashed arrows represent a modification caused by a matching operation (the at most $m+1$ children will each have size at most $2\eta$ fraction of its parent via $q-z \le 2 \eta z \le 2 \eta q$). The \textit{near-descendants} of a node is the set of nodes that can be reached by a downward path of exclusively solid edges; in \cref{fig:tree}, the near-descendants of $a$ are $\{b,c,g,h\}$. The \textit{near-height} of a node is the number of edges in the longest downward path from the node, consisting of only solid edges; in \cref{fig:tree}, the near-height of $a$ is $2$. The \textit{level} of a node is the number of dashed edges in the path from the root to this node; in \cref{fig:tree}, the level of $a$ is 0, and the level of $k$ is 2. 

The way we have chosen to modify $\sol'$ enforces some convenient invariants:

\begin{claim}\label{claim:midpoint}
    Whenever a distribution $p_i$ splits one of its distribution states in half, there was always strictly more than $2^m + \nicefrac{1}{8 \eta}$ coupling states on either side of the midpoint. 
\end{claim}
\begin{proof}
    If this were not true, then at least one coupling state must consist of a large fraction of the distribution state's mass, lower bounded by 
    \begin{equation*}
        \frac{1/2}{\nicefrac{1}{8 \eta} + 2^m} \ge \frac{1/2}{\nicefrac{1}{8 \eta} + \nicefrac{1}{10 \eta}} \ge 2 \eta, \rtag{using $m\ge 2$ and  $\eta \le \frac{1}{2^{3m} \cdot 4m} \implies \nicefrac{1}{10\eta} \ge 2^m$}
    \end{equation*}
    in which case this split would not have occurred. 
\end{proof}

\begin{claim}\label{claim:important-bound}
    Any distribution state contains at most $\nicefrac{1}{8 \eta}$ important coupling states, $2^{m-1}$ left-pierced coupling states, and $2^{m-1}$ right-pierced coupling states, at any time. 
    
    Further, this implies: (i) the midpoint of a split operation from some distribution will only ever split a coupling state it considers regular, (ii) any time a new distribution state $s_L$ is created by a split operation, it will start with exactly $\nicefrac{1}{8 \eta}$ important coupling states, (iii) if some distribution considers a coupling state pierced, it has near-height at most $m-1$, and (iv) any coupling state has near-height at most $m$. 
\end{claim}
\begin{proof}
    Suppose this main invariant has held so far. This is initially true because all coupling states start as regular. 

    Then, consider any split operation by a distribution $p_i$. We know from \cref{claim:midpoint} that there are strictly more than $2^m + \nicefrac{1}{8\eta}$ coupling states on either side of the midpoint. Since there are only at most $2^{m-1}$ left-pierced coupling states, they will be strictly to the left of the midpoint; also, since there are only at most $\nicefrac{1}{8 \eta}$ important coupling states and $2^{m-1}$ right-pierced coupling states, they will be strictly to the right of the midpoint. Hence, implication (i) holds: the midpoint must split a regular coupling state. Similarly, the new distribution state $s_L$ has strictly more than $2^m + \nicefrac{1}{8 \eta}$ coupling states (at most $2^{m-1}$ are left-pierced, the rest are regular), so implication (ii) holds, as there are enough coupling states to designate $\nicefrac{1}{8 \eta}$ as important.

    We should also examine the effect of this split operation from $p_i$ on some other distribution $p_j$. This never changes the number of important coupling states in $p_j$. While this may increase the number of pierced coupling states, we can bound the number of times this happens. Observe how once a distribution splits a coupling state, it considers all its near-descendants as pierced. By (i), this means the distribution will never split any of the near-descendants. With this in mind, consider in the tree any downward path consisting of only solid edges that starts at a pierced coupling state; this path must have at most $m-1$ edges, or it would directly contradict the previous observation (by pigeonhole, some distribution would be splitting a near-descendant). This proves (iii), which in turn implies (iv) since any downward path of solid edges will have some distribution that considers the second coupling state on the path as pierced. 
    
    All together, by (iii) we will conclude that the initial left-pierced coupling state will be split into at most $2^{m-1}$ left-pierced coupling states. This follows from how the current number of left-pierced coupling states is the number of leaf node coupling states reachable from the initial left-pierced coupling state via exclusively solid edges; this is at most 2 to the power of the coupling state's near-height, which is at most $m-1$ via (iii). The same argument holds for right-pierced coupling states.

    Observe how  matching operations never increase the number of pierced or important coupling states. This is enough to conclude our invariant holds.
\end{proof}

With these structural properties in hand, our plan will be to charge the immediate cost of splitting against the cost of later operations. For accounting purposes, we define an \textit{identifier} for each distribution state. For each identifier $r$, we also establish quantities $\operatorname{Cost}_r$, $\operatorname{Matched}_r$, and $\operatorname{Contain}_r$. At the beginning of the dynamic programming process, each distribution state is given a unique identifier and all $\operatorname{Cost}_r = \operatorname{Matched}_r = 0$.

Let us specify how identifiers are managed throughout the operations. In a matching operation, a distribution state of size $x_i$ is split into distribution states of size $z, a_i, b_i$. The distribution state of size $z$ will not have any identifier (because we ignore it from then onward), the size $a_i$ distribution state will be given a new identifier, and the size $b_i$ distribution state will inherit the original $x_i$ distribution state's identifier.

In a splitting operation, the right half $s_R$ will inherit the  identifier from $\ind_{i^*}$, and the left half $s_L$ will receive a new identifier $r'$. Let $x$ be the size of the coupling state that is split by this operation ($x=0$ if the split occurs between coupling states). Then, we set $\operatorname{Cost}_{r'} = x$.

\begin{claim}\label{claim:cost-bound}
    The total entropy increase from all splitting operations is at most $\sum_r \operatorname{Cost}_r$.
\end{claim}
\begin{proof}
    This follows from how the entropy increase from any individual splitting operation is at most $\operatorname{Cost}_{r'}$.
\end{proof}

We can now define the quantities $\operatorname{Matched}_r$ and $\operatorname{Contain}_r$:

\begin{definition}
    For any $r$, the quantity $\operatorname{Matched}_r$ represents the total mass matched inside $r$. More formally, for any matching operation with $z$ and distribution states $\ind_{1},\dots,\ind_{m}$, each $\operatorname{Matched}_{\ind_{i}}$ is increased by $z$.
\end{definition}

\begin{definition}
    $\operatorname{Contain}_r$ is equal to the size of the distribution state $r$ when the leftovers base case is eventually reached.
\end{definition}

More concretely, our plan is to bound the sum of all $\operatorname{Matched}_r,\operatorname{Contain}_r$, and then bound each $\operatorname{Cost}_r$ in relation to these quantities. The first part follows simply:

\begin{claim}\label{claim:match-contain-bound}
    $\sum_r \left(\operatorname{Matched}_r +  \operatorname{Contain}_r \right) = m $
\end{claim}
\begin{proof}
    For each distribution, all the mass is either matched or corresponds to the leftovers base case. Hence, the sum over $r$ corresponding to each $p_i$ is $1$, and the total sum is equal to $m$.
\end{proof}

We are now ready to bound each $\operatorname{Cost}_r$ with respect to $\operatorname{Matched}_r,\operatorname{Contain}_r$:

\begin{lemma}\label{lemma:cost-relate}
    For each identifier $r$, it holds that
    \begin{equation*}
        \operatorname{Cost}_r \le \eta  2^{m+4} \cdot (\operatorname{Matched}_r + \operatorname{Contain}_r)
    \end{equation*}
\end{lemma}
\begin{proof}
    Let $r$ be the identifier of a distribution state for some distribution $p_i$. If $r$ is an identifier for a distribution state where $\operatorname{Cost}_r=0$, then this lemma trivially holds. Otherwise, $r$ was created as the left half of a splitting operation. 
    
    We will show this lemma by analyzing what happens to important coupling states after the creation of $r$. By \cref{claim:important-bound}, we know the midpoint splits a regular coupling state, and when any $r$ is created it will have exactly $\nicefrac{1}{8 \eta}$ important coupling states. Since regular states are sorted in non-increasing order, the $\nicefrac{1}{8 \eta}$ important coupling states will all have size at least $\operatorname{Cost}_r$ when $r$ was created.

    Let us examine how these important coupling states change after $r$ is created. Because we know a split midpoint is always to the left of all important coupling states, we know that a split operation for $p_i$ will never move an important coupling state to a different distribution state of $p_i$. So, the only possible final outcomes for an important coupling state are either (i) it is involved in a matching, or (ii) it remains in $r$ until the leftover base case. 

    Before this final outcome, the only modification that may occur to an important coupling state would be caused by this coupling state being split by a midpoint of a splitting operation for some other $p_j$. When such an operation occurs, recall that we consider the larger half to remain the important coupling state, and that $p_j$ now considers this coupling state as pierced. Since \cref{claim:important-bound} argues the near-height of an important coupling state at the creation of $r$ is at most $m$, the important coupling state is at worst halved in size $m$ times, meaning it will have size at least $\operatorname{Cost}_r / 2^{m}$ at its final outcome. 

    For case (i), since we know for a matching $q-z \le 2 \eta z  \implies z \ge q/2$, implying the important coupling state will contribute at least $\operatorname{Cost}_r / 2^{m+1}$ to $\operatorname{Matched}_r$. Otherwise, for case (ii), we know the important coupling state will contribute at least $\operatorname{Cost}_r / 2^{m}$ to $\operatorname{Contain}_r$. Since there are $\nicefrac{1}{8 \eta}$ important coupling states, we reach our desired guarantee of
    \begin{equation*}
        \operatorname{Matched}_r + \operatorname{Contain}_r \ge \nicefrac{1}{8 \eta} \cdot \operatorname{Cost}_r / 2^{m+1} \implies \operatorname{Cost}_r \le \eta  2^{m+4} \cdot (\operatorname{Matched}_r + \operatorname{Contain}_r). \qedhere
    \end{equation*}
\end{proof}

\begin{corollary}
    The increase in entropy from splitting operations is at most $\eta m 2^{m+4}$.
\end{corollary}
\begin{proof}
    Using \cref{claim:match-contain-bound,lemma:cost-relate}, the total increase is bounded by
    \begin{equation*}
        \sum_r \operatorname{Cost}_r \le \sum_r \eta  2^{m+4} \cdot (\operatorname{Matched}_r + \operatorname{Contain}_r) = \eta m 2^{m+4}. \qedhere
    \end{equation*}
\end{proof}
\subsection{Leftover}

Lastly, we consider the leftover distribution states in the base case. The relevant quantity to bound is $\sum_{i=1}^m H(S_i^{\text{leftover}})$; crucially, each summand is upper bounded by the sum of the entropy values of the remaining coupling states in $\sol'$. We will analyze this last quantity by studying the process described in \cref{fig:tree}. Since the remaining leftover coupling states will all have size less than $\tau / \alpha$, it is sufficient to bound the entropy of the leaf coupling states with size less than $\tau / \alpha$ in the tree process of \cref{fig:tree}. Roughly, we should expect this quantity to be small, because the number of coupling states is only large when considering sufficiently large levels of the tree, but large levels of the tree have rapidly decaying mass (because increasing a level means reducing to $O(\eta)$ fraction of the parent's mass). We now execute this plan.

We aim to bound the entropy of leaf nodes in \cref{fig:tree} with size less than $\tau / \alpha$. Let us focus on a particular coupling state $\mathcal{C}$ in $\sol'$ at the beginning of the dynamic program. Let $\operatorname{Level}(\mathcal{C},i)$ denote the set of nodes with level $i$ in the tree with root $\mathcal{C}$, let $\operatorname{Small}(\mathcal{C})$ denote the \textit{leaves} in the tree with size less than $\tau / \alpha$, and let $\operatorname{size}(c)$ denote the size of some coupling state $c$.  We have argued how for each $j \in \{1,\dots,m\}$,
\begin{equation}
    H(S_j^{\text{leftover}}) \le \sum_\mathcal{C} \sum_{i = 0}^{\infty} \sum_{c \in  \operatorname{Level}(\mathcal{C},i) \cap \operatorname{Small}(\mathcal{C})} \phi(\operatorname{size}(c)).
\end{equation}
Hence, it follows that
\begin{equation}
    \sum_{j=1}^m H(S_j^{\text{leftover}}) \le m \cdot \sum_\mathcal{C} \sum_{i = 0}^{\infty} \sum_{c \in  \operatorname{Level}(\mathcal{C},i) \cap \operatorname{Small}(\mathcal{C})} \phi(\operatorname{size}(c)).\label{eq:leftover-sum}
\end{equation}

We now focus on bounding the key quantity in the right-hand side of \cref{eq:leftover-sum}:
\begin{lemma}\label{lemma:leftover-bound}
    For $\alpha = \eta^3/2$ and $\tau = \eta^9/n^6$, it holds that
    \begin{equation*}
        \sum_\mathcal{C} \sum_{i = 0}^{\infty} \sum_{c \in  \operatorname{Level}(\mathcal{C},i) \cap \operatorname{Small}(\mathcal{C})} \phi(\operatorname{size}(c)) \le \eta \cdot 1240m \cdot 2^{3m}.
    \end{equation*}
\end{lemma}
\begin{proof}
Let us examine how the number of states and the amount of mass changes throughout the tree. First, we may bound the mass of leaves in $\operatorname{Level}(\mathcal{C},i)$:

\begin{claim}\label{claim:tree-mass}
    $\sum_{c \in \operatorname{Level}(\mathcal{C},i) \cap \operatorname{Small}(\mathcal{C})}  \operatorname{size}(c) \le \operatorname{size}(\mathcal{C}) \cdot (2(m+1) \eta)^i$.
\end{claim}
\begin{proof}
    Let us define the \textit{top} nodes of $\operatorname{Level}(\mathcal{C},i)$ as all of the nodes in this level with no solid edge parent. Let us also define the \textit{bottom} nodes of $\operatorname{Level}(\mathcal{C},i)$ as all of the nodes in this level with no solid edge children. Observe how for any node with solid edge children, the sum of the two children's sizes is exactly the same as the original node's size; this is because splitting operations keep the total size unchanged. From this logic, we observe the top nodes of a level have the same total size as the bottom nodes of a level. Since $\operatorname{Level}(\mathcal{C},i) \cap \operatorname{Small}(\mathcal{C})$ is a subset of the nodes in the bottom set of level $i$, it is sufficient to provide the desired upper bound for the total size of the nodes in the top set of level $i$. For $i=0$, our desired bound holds by definition since the top set of level 0 is exactly $\mathcal{C}$. For $i>0$, we observe how the top set of level $i$ are all children of the bottom set of level $i-1$, via a matching operation. Since a matching operation reduces some coupling state to at most $m+1$ coupling states with at most $2 \eta$ fraction mass of the parent, then the total size of the top set of level $i$ is at most $(2(m+1)\eta)$ times the total size of the top set of level $i-1$. This inductively implies our desired bound. 
\end{proof}

 We can also bound the number of leaves in each level:

 \begin{claim}\label{claim:tree-states}
     $|\operatorname{Level}(\mathcal{C},i) \cap \operatorname{Small}(\mathcal{C})| \le 2^{m} \cdot ((m+1) \cdot 2^{m})^i$.
 \end{claim}
 \begin{proof}
     This follows from simple induction. For $i=0$, there are at most $2^m$ nodes in this level with no near-descendants, because the near-height of the root is bounded by $m$ via \cref{claim:important-bound}. 

     For $i>0$, consider all nodes with no near-descendants in level $i-1$ (we know there are at most $2^m \cdot ((m+1) \cdot 2^{m})^{i-1}$). Each of these will split into at most $m+1$ nodes through a matching operation. Then, the resulting nodes similarly have near-height bounded by $m$, so the number of nodes in level $i$ with no near-descendants is at most $2^{m} \cdot ((m+1) \cdot 2^{m})^i$. Since any leaf will have no near-descendants, this implies our claim.
 \end{proof}

\cref{claim:tree-states} and concavity of entropy will help us start our desired bound:

\begin{align*}
    &\sum_\mathcal{C} \sum_{i = 0}^{\infty} \sum_{c \in  \operatorname{Level}(\mathcal{C},i) \cap \operatorname{Small}(\mathcal{C})} \phi(\operatorname{size}(c)) \le \sum_\mathcal{C} \sum_{i = 0}^{\infty} \left( \left( \sum_{c \in  \operatorname{Level}(\mathcal{C},i) \cap \operatorname{Small}(\mathcal{C})} \operatorname{size}(c) \right) \cdot \log \left( \frac{|\operatorname{Level}(\mathcal{C},i) \cap \operatorname{Small}(\mathcal{C})|}{\sum_{c \in  \operatorname{Level}(\mathcal{C},i) \cap \operatorname{Small}(\mathcal{C})} \operatorname{size}(c) } \right) \right) \\
    &\le \sum_\mathcal{C} \sum_{i = 0}^{\infty} \left( \left( \sum_{c \in  \operatorname{Level}(\mathcal{C},i) \cap \operatorname{Small}(\mathcal{C})} \operatorname{size}(c) \right) \cdot \log \left( \frac{2^{m} \cdot ((m+1) \cdot 2^{m})^i}{\sum_{c \in  \operatorname{Level}(\mathcal{C},i) \cap \operatorname{Small}(\mathcal{C})} \operatorname{size}(c) } \right) \right) \intertext{Let $f_i(x) \triangleq x \log \left( \frac{2^m \cdot ((m+1) \cdot 2^m)^i}{x}\right)$, and let $g_i(x) \triangleq x \log \left( \frac{2^{2mi + m + 2}}{x}\right)$. Observe how for $x \in [0,1]$, $g_i$ is non-decreasing in $x$ and $f_i(x) \le g_i(x)$. We may use this to then bound:}
    &= \sum_\mathcal{C} \sum_{i = 0}^{\infty} f_i\left( \sum_{c \in  \operatorname{Level}(\mathcal{C},i) \cap \operatorname{Small}(\mathcal{C})} \operatorname{size}(c) \right) \\
    &\le \sum_\mathcal{C} \sum_{i = 0}^{\infty} g_i\left( \sum_{c \in  \operatorname{Level}(\mathcal{C},i) \cap \operatorname{Small}(\mathcal{C})} \operatorname{size}(c) \right) \intertext{We observe how the quantity $\sum_{c \in  \operatorname{Level}(\mathcal{C},i) \cap \operatorname{Small}(\mathcal{C})} \operatorname{size}(c)$ can be upper bounded by \cref{claim:tree-mass}, or by bounding the number of states with \cref{claim:tree-states} and using how each state has size at most $\tau / \alpha$. The former bound will be better for larger levels, and latter bound will be better for smaller levels. We invoke the latter for levels up to $j$ (an integer we will choose later), and the former for levels after $j$:}
    & \le \sum_\mathcal{C} \sum_{i = 0}^{j} g_i\left( (\tau / \alpha) \cdot |\operatorname{Level}(\mathcal{C},i) \cap \operatorname{Small}(\mathcal{C})| \right) + \sum_\mathcal{C} \sum_{i = j+1}^{\infty} g_i\left(  \operatorname{size}(\mathcal{C}) \cdot (2(m+1) \eta)^i \right)\\
    & \le \sum_\mathcal{C} \sum_{i = 0}^{j} g_i\left( (\tau / \alpha) \cdot 2^{m} \cdot ((m+1) \cdot 2^{m})^i \right) + \sum_\mathcal{C} \sum_{i = j+1}^{\infty} g_i\left( (2(m+1) \eta)^i   \right) \intertext{Since we are using monotonicity of $g_i$ for $x\in [0,1]$, the following holds as long as $(\tau / \alpha) \cdot 2^{2mi+m} \le 1$ (this will hold later when we have chosen parameters):}
    & \le \sum_\mathcal{C} \sum_{i = 0}^{j} g_i\left( (\tau / \alpha) \cdot 2^{2mi+m}\right) + \sum_\mathcal{C} \sum_{i = j+1}^{\infty} g_i\left( (4m \eta)^i   \right) \\
    & = \sum_\mathcal{C} \sum_{i = 0}^{j}(\tau / \alpha) \cdot 2^{2mi+m} \cdot \log \left( \frac{2^{2mi+m+2}}{(\tau / \alpha) \cdot 2^{2mi+m}} \right)  + \sum_\mathcal{C} \sum_{i = j+1}^{\infty} (4m\eta)^i \cdot \log \left( \frac{2^{2mi+m+2}}{(4m\eta)^i} \right) \\
    & \le \sum_\mathcal{C} \sum_{i = 0}^{j}(\tau / \alpha) \cdot 2^{2mi+m} \cdot (2 + \log(\alpha/\tau))  + \sum_\mathcal{C} \sum_{i = j+1}^{\infty} (4m\eta)^i \cdot i \cdot (3m + \log \left( 1 / \eta \right) ) \intertext{Since both series have consecutive ratios of at least $2$, using $4m \eta \le 1/4$, it holds:}
    & \le 2 \cdot \sum_\mathcal{C} \left(  (\tau / \alpha) \cdot 2^{2mj+m} \cdot (2 + \log(\alpha/\tau))  + (4m\eta)^{j+1} \cdot (j+1) \cdot (3m + \log \left( 1 / \eta \right) ) \right) \intertext{Since we will later choose a $j\ge 1$:}
    & \le 2 \cdot \sum_\mathcal{C} \left(  (\tau / \alpha) \cdot 2^{3mj} \cdot (2 + \log(\alpha/\tau))  + (4m\eta)^{j+1} \cdot (j+1) \cdot (3m + \log \left( 1 / \eta \right) ) \right) \intertext{What remains is to choose $\tau,j$ appropriately, and use a bound on the number of coupling states $\mathcal{C}$ in $\sol'$ at the start of the dynamic program. Lemma 5.1 of \cite{compton2023minimum} shows via a two-paragraph proof that for $m$ distributions over at most $n$ states, there always exists a minimum-entropy coupling with support size at most $nm - (m-1)$. If we choose $\sol$ to such a coupling, then \cref{claim:preprocessing-rounds} implies how after the preprocessing phase there will be at most $nm - (m-1) + nm \lceil \log(1/\tau) \rceil$ coupling states in $\sol'$: }
    & \le 2nm \cdot \lceil 1 + \log(1/\tau) \rceil \cdot  \left(  (\tau / \alpha) \cdot 2^{3mj} \cdot (2 + \log(\alpha/\tau))  + (4m\eta)^{j+1} \cdot (j+1) \cdot (3m + \log \left( 1 / \eta \right) ) \right) \\
    & \le 2nm \cdot \lceil 1 + \log(1/\tau) \rceil \cdot  \left(  (\tau / \alpha) \cdot 2^{3mj} \cdot (2 + \log(\alpha/\tau))  + (4m\eta)^{j} \cdot 4m\eta \cdot 4j\log \left( 1 / \eta \right) \right) \intertext{ Using  $\eta \le \frac{1}{2^{3m} \cdot 4m}$ and $x \log(1/x) \le 1$ for $x>0$:}
    & \le 2nm \cdot \lceil 1 + \log(1/\tau) \rceil \cdot  \left(  (\tau / \alpha) \cdot 2^{3mj} \cdot (2 + \log(\alpha/\tau))  + 4m \cdot 2^{-3mj} \cdot 4j \right) \intertext{We choose $j$ to be the smallest integer where $2^{-3mj} < (\eta/n)^3$. This choice of $j$ implies $1 \le j \le 3 \log(n/\eta)$:}
    & \le 2nm \cdot \lceil 1 + \log(1/\tau) \rceil \cdot  \left(  (\tau / \alpha) \cdot 2^{3m} \cdot (n/\eta)^3 \cdot (2 + \log(\alpha/\tau))  + 48m \cdot  \log(n/\eta) \cdot (\eta/n)^3 \right) \intertext{We recall $\alpha = \eta^3 / 2$, and choose $\tau = \eta^9 / n^6$:}
    & = 2nm \cdot \lceil 1 + \log(n^6 / \eta^9) \rceil \cdot  \left(  (2 \eta^6 / n^6) \cdot 2^{3m} \cdot (n/\eta)^3 \cdot (2 + \log(n^6 / (2\eta^6)))  + 48m \cdot  \log(n/\eta) \cdot (\eta/n)^3 \right)\\
    &\le  \eta \cdot  \left(  280m  \cdot 2^{3m} + 960 m^2\right) \le \eta \cdot 1240 m \cdot 2^{3m}
\end{align*}
Note how our earlier condition of $(\tau / \alpha) \cdot 2^{2mi + m} \le 1$ (for $i \le j$) is satisfied because the left-hand side is maximized when $i=j$, and then the quantity is upper bounded by $(\tau / \alpha) \cdot 2^{3mj} \le (n/\eta)^3 \cdot 2^{3m} \cdot (\tau / \alpha) = 2^{3m+1} \cdot (\eta/n)^3 < 1$. 
\end{proof}
Together, \cref{eq:leftover-sum,lemma:leftover-bound} bound the relevant quantity for the leftover case:
\begin{corollary}
    $\sum_{j=1}^m H(S_j^{\text{leftover}}) \le \eta \cdot 1240 m^2 \cdot 2^{3m}$.
\end{corollary}
\subsection{Concluding the analysis} 
Recall that our algorithm will output a coupling with entropy at most that of the dynamic programming value. Combining all these approximation error guarantees, we conclude how there exists a path in the dynamic program with value at most
\begin{align*}
    &H(\sol) + (H(\sol') - H(\sol)) + \sum_{i=1}^m H(S_i^{\text{leftover}}) \\
    &\le H(\sol) + (2\eta m) + (6 \eta \log(1/\eta)) + (\eta m 2^{m+4}) + (\eta \cdot 1240 m^2 \cdot 2^{3m}) \le H(\sol) + 6 \eta \log(1/\eta) + \eta \cdot 1242 m^2 \cdot 2^{3m}.
\end{align*}
Thus, we may conclude the total approximation error of \cref{alg:fast} is at most
\begin{equation}
    6 \eta \log(1/\eta) + \eta \cdot 1242 m^2 \cdot 2^{3m}. \label{eq:final-approx}
\end{equation}
If we invoke \cref{alg:fast} with $\eta$ defined in terms of $\eps$, then we obtain our main result:
\maintheorem*
\begin{proof}
    We will use \cref{alg:fast} with $\eta$ chosen to make \cref{eq:final-approx} sufficiently small. Suppose $\eta = \eps / (C \log(1/\eps))$ for some $C > 1$, then by \cref{eq:final-approx} the approximation error is bounded by:
    \begin{align*}
        6 \eta \log(1/\eta) + \eta \cdot 1242 m^2 \cdot 2^{3m} &\le \eps \cdot \left(\frac{12}{C} + \frac{6 \log(C)}{C} + \frac{1242m^2 \cdot 2^{3m}}{C}\right) \le \eps
    \end{align*}
    for $C = 3726 m^2 \cdot 2^{3m}$. Hence, if $\eta \le  \frac{\eps}{3726m^2 \cdot 2^{3m} \cdot \log(1/\eps)}$, then we achieve our desired approximation error and the value of $\eta$ satisfies the required condition that $\eta \le \frac{1}{2^{3m} \cdot 4m}$. Since we additionally require that $\eta$ is a power of $2$, we simply round down to the nearest power of $2$, and conclude this is a valid $\eta$ where $\eta \ge \frac{\eps}{7452m^2 \cdot 2^{3m} \cdot \log(1/\eps)}$. As we discussed earlier in \cref{sec:algo}, the runtime for \cref{alg:fast} is $(1/\tau)^{O(m \log(\nicefrac{1}{(\eta \alpha)}) / (\eta)^2)}$. Moreover, since $\alpha = \eta^3/2$ and $\tau = \eta^9/n^6$, we conclude the running time bound with our choice of $\eta$:
    \begin{align*}
        & (1/\tau)^{O(m \log(\nicefrac{1}{(\eta \alpha)}) / \eta^2)} \le n^{O(m \cdot \log(\nicefrac{1}{\eta}) \cdot  \log(\nicefrac{1}{(\eta \alpha)}) / (\eta^2)} \\
        & \le n^{O(m \cdot (\log(1/\eps) + m)^2 \cdot (\log^2(1/\eps) \cdot m^4 \cdot 2^{6m} / \eps^2))} \le n^{O(m^7 \cdot 2^{6m} \cdot \log^4(1/\eps)/\eps^2 )}. \quad \qedhere
    \end{align*}
\end{proof}
\section{Discussion}\label{sec:discuss}

In \cref{theorem:main-theorem}, we showed how there exists an $\eps$-approximate coupling algorithm with running time $n^{O(m^7 \cdot 2^{6m} \cdot \log^4(1/\eps)/\eps^2 )}$: proving there exists a PTAS for constant $m=O(1)$. We suspect that techniques similar to those in this work could likely yield an analogous result for a similar problem studied in \cite{cicalese2017bounds}, although we do not execute this plan ourselves.

The most pressing open problem is whether there exists such a PTAS for general $m$, or whether it is APX-hard. Designing a PTAS may require an entirely new algorithmic approach, as our dynamic programming style approach inherently has an exponential dependence in $m$. We do not conjecture whether there exists a PTAS or it is APX-hard.

More practically, it is also worth noting how the PTAS given by \cref{alg:fast} seems too slow for most applications. We expect there are some settings where even the exact, exponential-time algorithms in \cite{compton2023minimum} will perform faster. In this sense, a natural open problem is whether a faster (enough to be practical) PTAS exists for constant $m$. Moreover, tighter analyses of the greedy coupling algorithm \cite{kocaoglu2017entropic} still seem quite interesting, as it is perhaps the algorithm with best-known approximation guarantees among those with relatively practical running times.

\section*{Acknowledgments}
Thank you to Kristjan Greenewald, Dmitriy Katz, Murat Kocaoglu, and Benjamin Qi, for our time working together on \cite{compton2023minimum}. Thank you to the reviewers for their detailed feedback. This work was supported by the National Defense Science \& Engineering Graduate (NDSEG) Fellowship Program, Tselil Schramm's NSF CAREER Grant no. 2143246, and Gregory Valiant's Simons Foundation Investigator Award and NSF award AF-2341890.

\bibliography{ref}

@inproceedings{compton2023minimum,
  title={Minimum-entropy coupling approximation guarantees beyond the majorization barrier},
  author={Compton, Spencer and Katz, Dmitriy and Qi, Benjamin and Greenewald, Kristjan and Kocaoglu, Murat},
  booktitle={International Conference on Artificial Intelligence and Statistics},
  pages={10445--10469},
  year={2023},
  organization={PMLR}
}

@article{hochbaum1987using,
  title={Using dual approximation algorithms for scheduling problems theoretical and practical results},
  author={Hochbaum, Dorit S and Shmoys, David B},
  journal={Journal of the ACM (JACM)},
  volume={34},
  number={1},
  pages={144--162},
  year={1987},
  publisher={ACM New York, NY, USA}
}

@article{jansen2020closing,
  title={Closing the gap for makespan scheduling via sparsification techniques},
  author={Jansen, Klaus and Klein, Kim-Manuel and Verschae, Jos{\'e}},
  journal={Mathematics of Operations Research},
  volume={45},
  number={4},
  pages={1371--1392},
  year={2020},
  publisher={INFORMS}
}

@article{kovavcevic2012hardness,
  title={On the entropy of couplings},
  author={Kova{\v{c}}evi{\'c}, Mladen and Stanojevi{\'c}, Ivan and {\v{S}}enk, Vojin},
  journal={Information and Computation},
  volume={242},
  pages={369--382},
  year={2015},
  publisher={Elsevier}
}

@inproceedings{sokota2024computing,
  title={Computing Low-Entropy Couplings for Large-Support Distributions},
  author={Sokota, Samuel and Sam, Dylan and de Witt, Christian Schroeder and Compton, Spencer and Foerster, Jakob and Kolter, J Zico},
  booktitle={Uncertainty in Artificial Intelligence},
  pages={3279--3298},
  year={2024},
  organization={PMLR}
}

@inproceedings{sokota2022communicating,
  title={Communicating via markov decision processes},
  author={Sokota, Samuel and De Witt, Christian A Schroeder and Igl, Maximilian and Zintgraf, Luisa M and Torr, Philip and Strohmeier, Martin and Kolter, Zico and Whiteson, Shimon and Foerster, Jakob},
  booktitle={International Conference on Machine Learning},
  pages={20314--20328},
  year={2022},
  organization={PMLR}
}

@article{de2022perfectly,
  title={Perfectly secure steganography using minimum entropy coupling},
  author={de Witt, Christian Schroeder and Sokota, Samuel and Kolter, J Zico and Foerster, Jakob and Strohmeier, Martin},
  journal={arXiv preprint arXiv:2210.14889},
  year={2022}
}

@article{cicalese2019minimum,
  title={Minimum-entropy couplings and their applications},
  author={Cicalese, Ferdinando and Gargano, Luisa and Vaccaro, Ugo},
  journal={IEEE Transactions on Information Theory},
  volume={65},
  number={6},
  pages={3436--3451},
  year={2019},
  publisher={IEEE}
}

@inproceedings{kocaoglu2017entropic,
  title={Entropic causal inference},
  author={Kocaoglu, Murat and Dimakis, Alexandros and Vishwanath, Sriram and Hassibi, Babak},
  booktitle={Proceedings of the AAAI Conference on Artificial Intelligence},
  volume={31},
  number={1},
  year={2017}
}

@inproceedings{kocaoglu2017isit,
  title={Entropic causality and greedy minimum entropy coupling},
  author={Kocaoglu, Murat and Dimakis, Alexandros G and Vishwanath, Sriram and Hassibi, Babak},
  booktitle={2017 IEEE International Symposium on Information Theory (ISIT)},
  pages={1465--1469},
  year={2017},
  organization={IEEE}
}

@inproceedings{compton2022entropic,
  title={Entropic causal inference: Graph identifiability},
  author={Compton, Spencer and Greenewald, Kristjan and Katz, Dmitriy A and Kocaoglu, Murat},
  booktitle={International Conference on Machine Learning},
  pages={4311--4343},
  year={2022},
  organization={PMLR}
}

@article{compton2020entropic,
  title={Entropic causal inference: Identifiability and finite sample results},
  author={Compton, Spencer and Kocaoglu, Murat and Greenewald, Kristjan and Katz, Dmitriy},
  journal={Advances in Neural Information Processing Systems},
  volume={33},
  pages={14772--14782},
  year={2020}
}

@inproceedings{javidian2021quantum,
  title={Quantum entropic causal inference},
  author={Javidian, Mohammad Ali and Aggarwal, Vaneet and Bao, Fanglin and Jacob, Zubin},
  booktitle={Quantum Information and Measurement},
  pages={F2C--3},
  year={2021},
  organization={Optica Publishing Group}
}

@article{javidian2022quantum,
  title={Quantum causal inference in the presence of hidden common causes: An entropic approach},
  author={Javidian, Mohammad Ali and Aggarwal, Vaneet and Jacob, Zubin},
  journal={Physical Review A},
  volume={106},
  number={6},
  pages={062425},
  year={2022},
  publisher={APS}
}

@article{ebrahimi2024minimum,
  title={Minimum entropy coupling with bottleneck},
  author={Ebrahimi, Reza and Chen, Jun and Khisti, Ashish},
  journal={Advances in Neural Information Processing Systems},
  volume={37},
  pages={59655--59688},
  year={2024}
}

@inproceedings{bounoua2025learning,
  title={Learning to Match Unpaired Data with Minimum Entropy Coupling},
  author={Bounoua, Mustapha and Franzese, Giulio and Michiardi, Pietro},
  booktitle={Forty-second International Conference on Machine Learning},
year={2025}
}

@inproceedings{liang2023multimodal,
  title={Multimodal Learning Without Labeled Multimodal Data: Guarantees and Applications},
  author={Liang, Paul Pu and Ling, Chun Kai and Cheng, Yun and Obolenskiy, Alexander and Liu, Yudong and Pandey, Rohan and Wilf, Alex and Morency, Louis-Philippe and Salakhutdinov, Russ},
  booktitle={The Twelfth International Conference on Learning Representations},
  year={2024}
}

@inproceedings{chowdhury2025fundamental,
  title={Fundamental Limits of Perfect Concept Erasure},
  author={Chowdhury, Somnath Basu Roy and Dubey, Kumar Avinava and Beirami, Ahmad and Kidambi, Rahul and Monath, Nicholas and Ahmed, Amr and Chaturvedi, Snigdha},
  booktitle={International Conference on Artificial Intelligence and Statistics},
  pages={901--909},
  year={2025},
  organization={PMLR}
}

@inproceedings{zamani2024improving,
  title={Improving Achievability of Cache-Aided Private Variable-Length Coding with Zero Leakage},
  author={Zamani, Amirreza and Skoglund, Mikael},
  booktitle={2024 22nd International Symposium on Modeling and Optimization in Mobile, Ad Hoc, and Wireless Networks (WiOpt)},
  pages={218--224},
  year={2024},
  organization={IEEE}
}

@article{zamani2025variable,
  title={Variable-Length Coding with Zero and Non-Zero Privacy Leakage},
  author={Zamani, Amirreza and Skoglund, Mikael},
  journal={Entropy},
  volume={27},
  number={2},
  pages={124},
  year={2025},
  publisher={MDPI}
}

@inproceedings{zamani2025private,
  title={Private Variable-Length Coding with Sequential Encoder},
  author={Zamani, Amirreza and Oechtering, Tobias J and G{\"u}nd{\"u}z, Deniz and Skoglund, Mikael},
  booktitle={2025 IEEE Wireless Communications and Networking Conference (WCNC)},
  pages={1--6},
  year={2025},
  organization={IEEE}
}

@article{cardinal2008tight,
  title={Tight results on minimum entropy set cover},
  author={Cardinal, Jean and Fiorini, Samuel and Joret, Gwena{\"e}l},
  journal={Algorithmica},
  volume={51},
  number={1},
  pages={49--60},
  year={2008},
  publisher={Springer}
}

@article{vidyasagar2012metric,
  title={A metric between probability distributions on finite sets of different cardinalities and applications to order reduction},
  author={Vidyasagar, Mathukumalli},
  journal={IEEE Transactions on Automatic Control},
  volume={57},
  number={10},
  pages={2464--2477},
  year={2012},
  publisher={IEEE}
}

@inproceedings{cicalese2016approximating,
  title={Approximating probability distributions with short vectors, via information theoretic distance measures},
  author={Cicalese, Ferdinando and Gargano, Luisa and Vaccaro, Ugo},
  booktitle={2016 IEEE International Symposium on Information Theory (ISIT)},
  pages={1138--1142},
  year={2016},
  organization={IEEE}
}

@article{li2021efficient,
  title={Efficient approximate minimum entropy coupling of multiple probability distributions},
  author={Li, Cheuk Ting},
  journal={IEEE Transactions on Information Theory},
  volume={67},
  number={8},
  pages={5259--5268},
  year={2021},
  publisher={IEEE}
}

@inproceedings{rossi2019greedy,
  title={Greedy additive approximation algorithms for minimum-entropy coupling problem},
  author={Rossi, Massimiliano},
  booktitle={2019 IEEE International Symposium on Information Theory (ISIT)},
  pages={1127--1131},
  year={2019},
  organization={IEEE}
}

@inproceedings{compton2022tighter,
  title={A tighter approximation guarantee for greedy minimum entropy coupling},
  author={Compton, Spencer},
  booktitle={2022 IEEE International Symposium on Information Theory (ISIT)},
  pages={168--173},
  year={2022},
  organization={IEEE}
}

@inproceedings{shkel2023information,
  title={Information spectrum converse for minimum entropy couplings and functional representations},
  author={Shkel, Yanina Y and Yadav, Anuj Kumar},
  booktitle={2023 IEEE International Symposium on Information Theory (ISIT)},
  pages={66--71},
  year={2023},
  organization={IEEE}
}

@article{ma2025efficient,
  title={Efficient approximate Minimum-R{\'e}nyi entropy couplings},
  author={Ma, Ya-Jing and Wang, Feng and Wu, Xian-Yuan},
  journal={Discrete and Continuous Dynamical Systems-S},
  pages={0--0},
  year={2025},
  publisher={Discrete and Continuous Dynamical Systems-S}
}

@inproceedings{yadav2025information,
  title={Approximation Guarantees for Minimum Rényi Entropy Functional Representations},
  author={Yadav, Anuj Kumar and Shkel, Yanina Y},
  booktitle={2025 IEEE International Symposium on Information Theory (ISIT)},
  year={2025},
  organization={IEEE}
}

@article{leung1989bin,
  title={Bin packing with restricted piece sizes},
  author={Leung, Joseph YT},
  journal={Information Processing Letters},
  volume={31},
  number={3},
  pages={145--149},
  year={1989},
  publisher={Elsevier}
}

@article{hochba1997approximation,
  title={Approximation algorithms for NP-hard problems},
  author={Hochbaum, Dorit S},
  journal={ACM Sigact News},
  volume={28},
  number={2},
  pages={40--52},
  year={1997},
  publisher={ACM New York, NY, USA}
}

@inproceedings{alon1997approximation,
  title={Approximation schemes for scheduling},
  author={Alon, Noga and Azar, Yossi and Woeginger, Gerhard J and Yadid, Tal},
  booktitle={SODA},
  pages={493--500},
  year={1997}
}

@article{alon1998approximation,
  title={Approximation schemes for scheduling on parallel machines},
  author={Alon, Noga and Azar, Yossi and Woeginger, Gerhard J and Yadid, Tal},
  journal={Journal of Scheduling},
  volume={1},
  number={1},
  pages={55--66},
  year={1998},
  publisher={Wiley Online Library}
}

@article{jansen2010eptas,
  title={An EPTAS for scheduling jobs on uniform processors: using an MILP relaxation with a constant number of integral variables},
  author={Jansen, Klaus},
  journal={SIAM Journal on Discrete Mathematics},
  volume={24},
  number={2},
  pages={457--485},
  year={2010},
  publisher={SIAM}
}

@inproceedings{chen2014optimality,
  title={On the optimality of approximation schemes for the classical scheduling problem},
  author={Chen, Lin and Jansen, Klaus and Zhang, Guochuan},
  booktitle={Proceedings of the twenty-fifth annual ACM-SIAM symposium on Discrete algorithms},
  pages={657--668},
  year={2014},
  organization={SIAM}
}

@article{cicalese2017bounds,
  title={Bounds on the entropy of a function of a random variable and their applications},
  author={Cicalese, Ferdinando and Gargano, Luisa and Vaccaro, Ugo},
  journal={IEEE Transactions on Information Theory},
  volume={64},
  number={4},
  pages={2220--2230},
  year={2017},
  publisher={IEEE}
}

@article{allender2009complexity,
  title={On the complexity of numerical analysis},
  author={Allender, Eric and B{\"u}rgisser, Peter and Kjeldgaard-Pedersen, Johan and Miltersen, Peter Bro},
  journal={SIAM Journal on Computing},
  volume={38},
  number={5},
  pages={1987--2006},
  year={2009},
  publisher={SIAM}
}

@article{etessami2010complexity,
  title={On the complexity of Nash equilibria and other fixed points},
  author={Etessami, Kousha and Yannakakis, Mihalis},
  journal={SIAM Journal on Computing},
  volume={39},
  number={6},
  pages={2531--2597},
  year={2010},
  publisher={SIAM}
}

@article{kayal2012sum,
  title={On the sum of square roots of polynomials and related problems},
  author={Kayal, Neeraj and Saha, Chandan},
  journal={ACM Transactions on Computation Theory (TOCT)},
  volume={4},
  number={4},
  pages={1--15},
  year={2012},
  publisher={ACM New York, NY, USA}
}

\appendix

\subsection{Remarks on finite-precision arithmetic}\label{appendix:precise}

As far as we know, it is open whether the simple expression $\phi(x)=x \log(1/x)$ can be exactly computed in polynomial time (e.g. see Remark 3.5 of \cite{kovavcevic2012hardness}). A similar phenomenon, which is better studied, is how we do not know whether it is possible to exactly compute the sum of square roots $\sum_{i} \sqrt{x_i}$ in polynomial time \cite{allender2009complexity,etessami2010complexity,kayal2012sum}. 

We will assume the input size of distribution states are rationals that may be represented within $L$ bits each. For our purposes, any instance where entropy is computed in this algorithm should actually be replaced with a subroutine that approximately computes $\phi(x)$. Observe that any coupling considered by the dynamic program ultimately consists of at most $O(m ( 1/\tau + n))$ coupling states, so if each entropy term is calculated up to additive accuracy $\delta = \frac{\eps}{C(m (1/\tau + n))}$ for a sufficiently large $C$, then the dynamic program value will be accurate (compared to exactly computing the entropy) up to an additive $\eps$ term. For a concrete analysis, this means if you invoke the algorithm in our work with $\eps/2$, and compute entropy up to $\delta = \frac{\eps}{C(m (1/\tau + n))}$ additive precision for large enough $C$, then the algorithm will yield an $\eps$-additive guarantee.

Computing the $\phi(x)$ terms up to additive accuracy $\delta$ is relatively standard: for a rough example of one approach, you can estimate $0$ for any $x \le \delta^2$, and otherwise you can approximate $-x \log(x)$ sufficiently well by approximating $\log(x)$ with a $\poly(1/\delta)$-order Taylor expansion of $\log(1-z)$. We avoid some details, as it seems besides the point of this work to get into specifics what can be done in machine word operations, but we note that this entropy-term approximation does not dominate the earlier running time bound. If each of the approximate entropy evaluations take $\poly(1/\delta, L, m, 1/\eps, n)$ time, this multiplicative increase in runtime is still dominated by the final runtime bound of $n^{O(m^7 \cdot 2^{6m} \cdot \log^4(1/\eps)/\eps^2 )}$ (we suppress the polynomial in $L$ bit-complexity factors in our runtime analysis, as is standard). Throughout the algorithm, we will also store values of $\G$, which take the form  $2^{-i}\eta j$, where $i \in \mathbb{Z}$ and $j \in \{1,\dots, \nicefrac{2}{\eta^2}-1\}$; we represent the values of $\G$ by the integer pair $(i,j)$. Note that the relevant values of $i$ are all $O(\log(1/\tau))$, and all operations involving $\G$ can be completed with the same negligible multiplicative increase in time. For example, comparing two elements of $\G$, $(i_1,j_1)$ and $(i_2,j_2)$, only requires comparing whether $j_1 > j_2 \cdot 2^{i_1 - i_2}$.

\end{document}